\documentclass[a4paper,onecolumn,11pt,unpublished]{quantumarticle}
\pdfoutput=1

\usepackage[utf8]{inputenc}

\usepackage{amsmath}
\usepackage{thmtools}
\usepackage{amssymb}
\usepackage{bbold}
\usepackage{comment}
\usepackage{graphicx}
\usepackage{scalerel}
\usepackage{braket}
\usepackage{leftidx}
\usepackage[title]{appendix}
\usepackage{tikz}
\usepackage{tikzit}
\usepackage{tikz-cd}
\usepackage{cite}
\usepackage{geometry}
\usepackage{mathtools}
\usepackage{ stmaryrd }
\usepackage{circuitikz}
\usepackage{textcomp}
\usepackage{accents}
\usepackage{comment}
\usepackage{ mathrsfs }
\usepackage{authblk}
\usepackage{breakurl,amsmath,amsthm,amssymb,xspace,xypic,enumerate,color}
\usepackage{float}
\usepackage{thm-restate}
\newcommand{\kett}[1]{|#1)}
\newcommand{\braa}[1]{(#1|}

\tikzcdset{scale cd/.style={every label/.append style={scale=#1},
    cells={nodes={scale=#1}}}}

\usetikzlibrary{shapes.multipart}

\tikzstyle{bwSpider}=[
       rectangle split,
       rectangle split parts=2,
       rectangle split part fill={black,white},
 minimum size=3.6 mm, inner sep=-2mm, draw=black,scale=0.5,rounded corners=0.8 mm
       ]
 \tikzstyle{wbSpider}=[
       rectangle split,
       rectangle split parts=2,
       rectangle split part fill={white,black},
 minimum size=3.6 mm, inner sep=-2mm, draw=black,scale=0.5,rounded corners=0.8 mm
       ]
\tikzstyle{cWire}=[densely dotted, thick]
\tikzstyle{env}=[copoint,regular polygon rotate=0,minimum width=0.2cm, fill=black]

\tikzstyle{probs}=[shape=semicircle,fill=white,draw=black,shape border rotate=180,minimum width=1.2cm]

\tikzstyle{every picture}=[baseline=-0.25em,scale=0.5]
\tikzstyle{dotpic}=[] 
\tikzstyle{diredges}=[every to/.style={diredge}]
\tikzstyle{math matrix}=[matrix of math nodes,left delimiter=(,right delimiter=),inner sep=2pt,column sep=1em,row sep=0.5em,nodes={inner sep=0pt},text height=1.5ex, text depth=0.25ex]

\tikzstyle{inline text}=[text height=1.5ex, text depth=0.25ex,yshift=0.5mm]
\tikzstyle{label}=[font=\footnotesize,text height=1.5ex, text depth=0.25ex,yshift=0.5mm]
\tikzstyle{left label}=[label,anchor=east,xshift=1.5mm]
\tikzstyle{right label}=[label,anchor=west,xshift=-1.5mm]

\tikzstyle{braceedge}=[decorate,decoration={brace,amplitude=2mm,raise=-1mm}]
\tikzstyle{small braceedge}=[decorate,decoration={brace,amplitude=1mm,raise=-1mm}]

\tikzstyle{doubled}=[line width=1.6pt] 
\tikzstyle{boldedge}=[doubled,shorten <=-0.17mm,shorten >=-0.17mm]
\tikzstyle{boldedgegray}=[doubled,gray,shorten <=-0.17mm,shorten >=-0.17mm]
\tikzstyle{singleedgegray}=[gray]

\tikzstyle{semidoubled}=[line width=1.4pt] 
\tikzstyle{semiboldedgegray}=[semidoubled,gray,shorten <=-0.17mm,shorten >=-0.17mm]

\tikzstyle{boxedge}=[semiboldedgegray]

\tikzstyle{dottededge}=[dashed,shorten <=-0.02mm,shorten >=-0.02mm]

\tikzstyle{boldedgedashed}=[very thick,dashed,shorten <=-0.17mm,shorten >=-0.17mm]
\tikzstyle{vboldedgedashed}=[doubled,dashed,shorten <=-0.17mm,shorten >=-0.17mm]
\tikzstyle{left hook arrow}=[left hook-latex]
\tikzstyle{right hook arrow}=[right hook-latex]
\tikzstyle{sembracket}=[line width=0.5pt,shorten <=-0.07mm,shorten >=-0.07mm]

\tikzstyle{causal edge}=[->,thick,gray]
\tikzstyle{causal nondir}=[thick,gray]
\tikzstyle{timeline}=[thick,gray, dashed]

\tikzstyle{cedge}=[<->,thick,gray!70!white]

\tikzstyle{empty diagram}=[draw=gray!40!white,dashed,shape=rectangle,minimum width=1cm,minimum height=1cm]
\tikzstyle{empty diagram small}=[draw=gray!50!white,dashed,shape=rectangle,minimum width=0.6cm,minimum height=0.5cm]

\tikzstyle{dot}=[inner sep=0mm,minimum width=2mm,minimum height=2mm,draw,shape=circle]

\tikzstyle{phase dot}=[pdot,phase dimensions]
\tikzstyle{wphase dot}=[dot, phase dimensions]

\tikzstyle{leak}=[white dot, shape=regular polygon, minimum size=300mm, regular polygon sides=3, outer sep=-0.2mm, regular polygon rotate=270]
\tikzstyle{preleak}=[trapezium, trapezium angle=67.5, draw, inner sep=0.1pt, outer sep=0pt, minimum height=2mm, minimum width=4pt,rotate=270]
\tikzstyle{proj}=[white dot, shape=regular polygon, minimum size=3.3 mm, regular polygon sides=4, outer sep=-0.2mm]
\tikzstyle{Vleak}=[white dot, shape=regular polygon, minimum size=3.3 mm, regular polygon sides=3, outer sep=-0.2mm, regular polygon rotate=90]
\tikzstyle{dleak}=[white dot, line width=1.6pt, shape=regular polygon, minimum size=3.3 mm, regular polygon sides=3, outer sep=-0.2mm, regular polygon rotate=270]

\tikzstyle{Wsquare}=[white dot, shape=regular polygon, rounded corners=0.8 mm, minimum size=3.3 mm, regular polygon sides=3, outer sep=-0.2mm]
\tikzstyle{Wsquareadj}=[white dot, shape=regular polygon, rounded corners=0.8 mm, minimum size=3.3 mm, regular polygon sides=3, outer sep=-0.2mm, regular polygon rotate=180]
\tikzstyle{ddot}=[inner sep=0mm, doubled, minimum width=2.5mm,minimum height=2.5mm,draw,shape=circle]

\tikzstyle{black dot}=[dot,fill=black]
\tikzstyle{white dot}=[dot,fill=white,,text depth=-0.2mm]
\tikzstyle{white Wsquare}=[Wsquare,fill=gray,,text depth=-0.2mm]
\tikzstyle{white Wsquareadj}=[Wsquareadj,fill=white,,text depth=-0.2mm]
\tikzstyle{green dot}=[white dot] 
\tikzstyle{gray dot}=[dot,fill=gray!40!white,,text depth=-0.2mm]
\tikzstyle{red dot}=[gray dot] 

\tikzstyle{black ddot}=[ddot,fill=black]
\tikzstyle{white ddot}=[ddot,fill=white]
\tikzstyle{gray ddot}=[ddot,fill=gray!40!white]

\tikzstyle{gray edge}=[gray!60!white]

\tikzstyle{small dot}=[inner sep=0.5mm,minimum width=0pt,minimum height=0pt,draw,shape=circle]

\tikzstyle{small black dot}=[small dot,fill=black]
\tikzstyle{small white dot}=[small dot,fill=white]
\tikzstyle{small gray dot}=[small dot,fill=gray!40!white]

\tikzstyle{causal dot}=[inner sep=0.4mm,minimum width=0pt,minimum height=0pt,draw=white,shape=circle,fill=gray!40!white]

\tikzstyle{phase dimensions}=[minimum size=5mm,font=\footnotesize,rectangle,rounded corners=2.5mm,inner sep=0.2mm,outer sep=-2mm]
\tikzstyle{dphase dimensions}=[minimum size=5mm,font=\footnotesize,rectangle,rounded corners=2.5mm,inner sep=0.2mm,outer sep=-2mm]

\tikzstyle{white phase dot}=[dot,fill=white,phase dimensions]
\tikzstyle{white phase ddot}=[ddot,fill=white,dphase dimensions]

\tikzstyle{white rect ddot}=[draw=black,fill=white,doubled,minimum size=5mm,font=\footnotesize,rectangle,rounded corners=2.5mm,inner sep=0.2mm]
\tikzstyle{gray rect ddot}=[draw=black,fill=gray!40!white,doubled,minimum size=6mm,font=\footnotesize,rectangle,rounded corners=3mm]

\tikzstyle{gray phase dot}=[dot,fill=gray!40!white,phase dimensions]
\tikzstyle{gray phase ddot}=[ddot,fill=gray!40!white,dphase dimensions]
\tikzstyle{grey phase dot}=[gray phase dot]
\tikzstyle{grey phase ddot}=[gray phase ddot]

\tikzstyle{small phase dimensions}=[minimum size=4mm,font=\tiny,rectangle,rounded corners=2mm,inner sep=0.2mm,outer sep=-2mm]
\tikzstyle{small dphase dimensions}=[minimum size=4mm,font=\tiny,rectangle,rounded corners=2mm,inner sep=0.2mm,outer sep=-2mm]

\tikzstyle{small gray phase dot}=[dot,fill=gray!40!white,small phase dimensions]
\tikzstyle{small gray phase ddot}=[ddot,fill=gray!40!white,small dphase dimensions]

\tikzstyle{small map}=[draw,shape=rectangle,minimum height=4mm,minimum width=4mm,fill=white]

\tikzstyle{cnot}=[fill=white,shape=circle,inner sep=-1.4pt]

\tikzstyle{asym hadamard}=[fill=white,draw,shape=NEbox,inner sep=0.6mm,font=\footnotesize,minimum height=4mm]
\tikzstyle{asym hadamard conj}=[fill=white,draw,shape=NWbox,inner sep=0.6mm,font=\footnotesize,minimum height=4mm]
\tikzstyle{asym hadamard dag}=[fill=white,draw,shape=SEbox,inner sep=0.6mm,font=\footnotesize,minimum height=4mm]

\tikzstyle{hadamard}=[fill=white,draw,inner sep=0.6mm,font=\footnotesize,minimum height=4mm,minimum width=4mm]
\tikzstyle{small hadamard}=[fill=white,draw,inner sep=0.6mm,minimum height=1.5mm,minimum width=1.5mm]
\tikzstyle{small hadamard rotate}=[small hadamard,rotate=45]
\tikzstyle{dhadamard}=[hadamard,doubled]
\tikzstyle{small dhadamard}=[small hadamard,doubled]
\tikzstyle{small dhadamard rotate}=[small hadamard rotate,doubled]
\tikzstyle{antipode}=[white dot,inner sep=0.3mm,font=\footnotesize]

\tikzstyle{scalar}=[diamond,draw,inner sep=0.5pt,font=\small]
\tikzstyle{dscalar}=[diamond,doubled, draw,inner sep=0.5pt,font=\small]

\tikzstyle{small box}=[rectangle,inline text,fill=white,draw,minimum height=5mm,yshift=-0.5mm,minimum width=5mm,font=\small]
\tikzstyle{small gray box}=[small box,fill=gray!30]
\tikzstyle{medium box}=[rectangle,inline text,fill=white,draw,minimum height=5mm,yshift=-0.5mm,minimum width=10mm,font=\small]
\tikzstyle{square box}=[small box] 
\tikzstyle{medium gray box}=[small box,fill=gray!30]
\tikzstyle{semilarge box}=[rectangle,inline text,fill=white,draw,minimum height=5mm,yshift=-0.5mm,minimum width=12.5mm,font=\small]
\tikzstyle{large box}=[rectangle,inline text,fill=white,draw,minimum height=5mm,yshift=-0.5mm,minimum width=15mm,font=\small]
\tikzstyle{large gray box}=[small box,fill=gray!30]

\tikzstyle{Bayes box}=[rectangle,fill=black,draw, minimum height=3mm, minimum width=3mm]

\tikzstyle{gray square point}=[small box,fill=gray!50]

\tikzstyle{dphase box white}=[dhadamard]
\tikzstyle{dphase box gray}=[dhadamard,fill=gray!50!white]
\tikzstyle{phase box white}=[hadamard]
\tikzstyle{phase box gray}=[hadamard,fill=gray!50!white]

\tikzstyle{point}=[regular polygon,regular polygon sides=3,draw,scale=0.75,inner sep=-0.5pt,minimum width=9mm,fill=white,regular polygon rotate=180]
\tikzstyle{point nosep}=[regular polygon,regular polygon sides=3,draw,scale=0.75,inner sep=-2pt,minimum width=9mm,fill=white,regular polygon rotate=180]
\tikzstyle{copoint}=[regular polygon,regular polygon sides=3,draw,scale=0.75,inner sep=-0.5pt,minimum width=9mm,fill=white]
\tikzstyle{dpoint}=[point,doubled]
\tikzstyle{dcopoint}=[copoint,doubled]

\tikzstyle{pointgrow}=[shape=cornerpoint,kpoint common,scale=0.75,inner sep=3pt]
\tikzstyle{pointgrow dag}=[shape=cornercopoint,kpoint common,scale=0.75,inner sep=3pt]

\tikzstyle{wide copoint}=[fill=white,draw,shape=isosceles triangle,shape border rotate=90,isosceles triangle stretches=true,inner sep=0pt,minimum width=1.5cm,minimum height=6.12mm]
\tikzstyle{wide point}=[fill=white,draw,shape=isosceles triangle,shape border rotate=-90,isosceles triangle stretches=true,inner sep=0pt,minimum width=1.5cm,minimum height=6.12mm,yshift=-0.0mm]
\tikzstyle{wide point plus}=[fill=white,draw,shape=isosceles triangle,shape border rotate=-90,isosceles triangle stretches=true,inner sep=0pt,minimum width=1.74cm,minimum height=7mm,yshift=-0.0mm]

\tikzstyle{wide dpoint}=[fill=white,doubled,draw,shape=isosceles triangle,shape border rotate=-90,isosceles triangle stretches=true,inner sep=0pt,minimum width=1.5cm,minimum height=6.12mm,yshift=-0.0mm]

\tikzstyle{tinypoint}=[regular polygon,regular polygon sides=3,draw,scale=0.55,inner sep=-0.15pt,minimum width=6mm,fill=white,regular polygon rotate=180]

\tikzstyle{white point}=[point]
\tikzstyle{white dpoint}=[dpoint]
\tikzstyle{green point}=[white point] 
\tikzstyle{white copoint}=[copoint]
\tikzstyle{gray point}=[point,fill=gray!40!white]
\tikzstyle{gray dpoint}=[gray point,doubled]
\tikzstyle{red point}=[gray point] 
\tikzstyle{gray copoint}=[copoint,fill=gray!40!white]
\tikzstyle{gray dcopoint}=[gray copoint,doubled]

\tikzstyle{white point guide}=[regular polygon,regular polygon sides=3,font=\scriptsize,draw,scale=0.65,inner sep=-0.5pt,minimum width=9mm,fill=white,regular polygon rotate=180]

\tikzstyle{black point}=[point,fill=black,font=\color{white}]
\tikzstyle{black copoint}=[copoint,fill=black,font=\color{white}]

\tikzstyle{tiny gray point}=[tinypoint,fill=gray!40!white]

\tikzstyle{diredge}=[->]
\tikzstyle{ddiredge}=[<->]
\tikzstyle{rdiredge}=[<-]
\tikzstyle{thickdiredge}=[->, very thick]
\tikzstyle{pointer edge}=[->,very thick,gray]
\tikzstyle{pointer edge part}=[very thick,gray]
\tikzstyle{dashed edge}=[dashed]
\tikzstyle{thick dashed edge}=[very thick,dashed]
\tikzstyle{thick gray dashed edge}=[thick dashed edge,gray!40]
\tikzstyle{thick map edge}=[very thick,|->]

\makeatletter
\newcommand{\boxshape}[3]{%
\pgfdeclareshape{#1}{
\inheritsavedanchors[from=rectangle] 
\inheritanchorborder[from=rectangle]
\inheritanchor[from=rectangle]{center}
\inheritanchor[from=rectangle]{north}
\inheritanchor[from=rectangle]{south}
\inheritanchor[from=rectangle]{west}
\inheritanchor[from=rectangle]{east}
\backgroundpath{
\southwest \pgf@xa=\pgf@x \pgf@ya=\pgf@y
\northeast \pgf@xb=\pgf@x \pgf@yb=\pgf@y

\@tempdima=#2
\@tempdimb=#3

\pgfpathmoveto{\pgfpoint{\pgf@xa - 5pt + \@tempdima}{\pgf@ya}}
\pgfpathlineto{\pgfpoint{\pgf@xa - 5pt - \@tempdima}{\pgf@yb}}
\pgfpathlineto{\pgfpoint{\pgf@xb + 5pt + \@tempdimb}{\pgf@yb}}
\pgfpathlineto{\pgfpoint{\pgf@xb + 5pt - \@tempdimb}{\pgf@ya}}
\pgfpathlineto{\pgfpoint{\pgf@xa - 5pt + \@tempdima}{\pgf@ya}}
\pgfpathclose
}
}}

\boxshape{NEbox}{0pt}{5pt}
\boxshape{SEbox}{0pt}{-5pt}
\boxshape{NWbox}{5pt}{0pt}
\boxshape{SWbox}{-5pt}{0pt}
\boxshape{EBox}{-3pt}{3pt}
\boxshape{WBox}{3pt}{-3pt}
\makeatother

\tikzstyle{cloud}=[shape=cloud,draw,minimum width=1.5cm,minimum height=1.5cm]

\tikzstyle{map}=[draw,shape=NEbox,inner sep=2pt,minimum height=6mm,fill=white]
\tikzstyle{dashedmap}=[draw,dashed,shape=NEbox,inner sep=2pt,minimum height=6mm,fill=white]
\tikzstyle{mapdag}=[draw,shape=SEbox,inner sep=2pt,minimum height=6mm,fill=white]
\tikzstyle{mapadj}=[draw,shape=SEbox,inner sep=2pt,minimum height=6mm,fill=white]
\tikzstyle{maptrans}=[draw,shape=SWbox,inner sep=2pt,minimum height=6mm,fill=white]
\tikzstyle{mapconj}=[draw,shape=NWbox,inner sep=2pt,minimum height=6mm,fill=white]

\tikzstyle{medium map}=[draw,shape=NEbox,inner sep=2pt,minimum height=6mm,fill=white,minimum width=7mm]
\tikzstyle{medium map dag}=[draw,shape=SEbox,inner sep=2pt,minimum height=6mm,fill=white,minimum width=7mm]
\tikzstyle{medium map adj}=[draw,shape=SEbox,inner sep=2pt,minimum height=6mm,fill=white,minimum width=7mm]
\tikzstyle{medium map trans}=[draw,shape=SWbox,inner sep=2pt,minimum height=6mm,fill=white,minimum width=7mm]
\tikzstyle{medium map conj}=[draw,shape=NWbox,inner sep=2pt,minimum height=6mm,fill=white,minimum width=7mm]
\tikzstyle{semilarge map}=[draw,shape=NEbox,inner sep=2pt,minimum height=6mm,fill=white,minimum width=9.5mm]
\tikzstyle{semilarge map trans}=[draw,shape=SWbox,inner sep=2pt,minimum height=6mm,fill=white,minimum width=9.5mm]
\tikzstyle{semilarge map adj}=[draw,shape=SEbox,inner sep=2pt,minimum height=6mm,fill=white,minimum width=9.5mm]
\tikzstyle{semilarge map dag}=[draw,shape=SEbox,inner sep=2pt,minimum height=6mm,fill=white,minimum width=9.5mm]
\tikzstyle{semilarge map conj}=[draw,shape=NWbox,inner sep=2pt,minimum height=6mm,fill=white,minimum width=9.5mm]
\tikzstyle{large map}=[draw,shape=NEbox,inner sep=2pt,minimum height=6mm,fill=white,minimum width=12mm]
\tikzstyle{large map conj}=[draw,shape=NWbox,inner sep=2pt,minimum height=6mm,fill=white,minimum width=12mm]
\tikzstyle{very large map}=[draw,shape=NEbox,inner sep=2pt,minimum height=6mm,fill=white,minimum width=17mm]

\tikzstyle{medium dmap}=[draw,doubled,shape=NEbox,inner sep=2pt,minimum height=6mm,fill=white,minimum width=7mm]
\tikzstyle{medium dmap dag}=[draw,doubled,shape=SEbox,inner sep=2pt,minimum height=6mm,fill=white,minimum width=7mm]
\tikzstyle{medium dmap adj}=[draw,doubled,shape=SEbox,inner sep=2pt,minimum height=6mm,fill=white,minimum width=7mm]
\tikzstyle{medium dmap trans}=[draw,doubled,shape=SWbox,inner sep=2pt,minimum height=6mm,fill=white,minimum width=7mm]
\tikzstyle{medium dmap conj}=[draw,doubled,shape=NWbox,inner sep=2pt,minimum height=6mm,fill=white,minimum width=7mm]
\tikzstyle{semilarge dmap}=[draw,doubled,shape=NEbox,inner sep=2pt,minimum height=6mm,fill=white,minimum width=9.5mm]
\tikzstyle{semilarge dmap trans}=[draw,doubled,shape=SWbox,inner sep=2pt,minimum height=6mm,fill=white,minimum width=9.5mm]
\tikzstyle{semilarge dmap adj}=[draw,doubled,shape=SEbox,inner sep=2pt,minimum height=6mm,fill=white,minimum width=9.5mm]
\tikzstyle{semilarge dmap dag}=[draw,doubled,shape=SEbox,inner sep=2pt,minimum height=6mm,fill=white,minimum width=9.5mm]
\tikzstyle{semilarge dmap conj}=[draw,doubled,shape=NWbox,inner sep=2pt,minimum height=6mm,fill=white,minimum width=9.5mm]
\tikzstyle{large dmap}=[draw,doubled,shape=NEbox,inner sep=2pt,minimum height=6mm,fill=white,minimum width=12mm]
\tikzstyle{large dmap conj}=[draw,doubled,shape=NWbox,inner sep=2pt,minimum height=6mm,fill=white,minimum width=12mm]
\tikzstyle{large dmap trans}=[draw,doubled,shape=SWbox,inner sep=2pt,minimum height=6mm,fill=white,minimum width=12mm]
\tikzstyle{large dmap adj}=[draw,doubled,shape=SEbox,inner sep=2pt,minimum height=6mm,fill=white,minimum width=12mm]
\tikzstyle{large dmap dag}=[draw,doubled,shape=SEbox,inner sep=2pt,minimum height=6mm,fill=white,minimum width=12mm]
\tikzstyle{very large dmap}=[draw,doubled,shape=NEbox,inner sep=2pt,minimum height=6mm,fill=white,minimum width=19.5mm]

\tikzstyle{muxbox}=[draw,shape=rectangle,minimum height=3mm,minimum width=3mm,fill=white]
\tikzstyle{dmuxbox}=[muxbox,doubled]

\tikzstyle{box}=[draw,shape=rectangle,inner sep=2pt,minimum height=6mm,minimum width=6mm,fill=white]
\tikzstyle{dbox}=[draw,doubled,shape=rectangle,inner sep=2pt,minimum height=6mm,minimum width=6mm,fill=white]
\tikzstyle{dmap}=[draw,doubled,shape=NEbox,inner sep=2pt,minimum height=6mm,fill=white]
\tikzstyle{dmapdag}=[draw,doubled,shape=SEbox,inner sep=2pt,minimum height=6mm,fill=white]
\tikzstyle{dmapadj}=[draw,doubled,shape=SEbox,inner sep=2pt,minimum height=6mm,fill=white]
\tikzstyle{dmaptrans}=[draw,doubled,shape=SWbox,inner sep=2pt,minimum height=6mm,fill=white]
\tikzstyle{dmapconj}=[draw,doubled,shape=NWbox,inner sep=2pt,minimum height=6mm,fill=white]

\tikzstyle{ddmap}=[draw,doubled,dashed,shape=NEbox,inner sep=2pt,minimum height=6mm,fill=white]
\tikzstyle{ddmapdag}=[draw,doubled,dashed,shape=SEbox,inner sep=2pt,minimum height=6mm,fill=white]
\tikzstyle{ddmapadj}=[draw,doubled,dashed,shape=SEbox,inner sep=2pt,minimum height=6mm,fill=white]
\tikzstyle{ddmaptrans}=[draw,doubled,dashed,shape=SWbox,inner sep=2pt,minimum height=6mm,fill=white]
\tikzstyle{ddmapconj}=[draw,doubled,dashed,shape=NWbox,inner sep=2pt,minimum height=6mm,fill=white]

\boxshape{sNEbox}{0pt}{3pt}
\boxshape{sSEbox}{0pt}{-3pt}
\boxshape{sNWbox}{3pt}{0pt}
\boxshape{sSWbox}{-3pt}{0pt}
\tikzstyle{smap}=[draw,shape=sNEbox,fill=white]
\tikzstyle{smapdag}=[draw,shape=sSEbox,fill=white]
\tikzstyle{smapadj}=[draw,shape=sSEbox,fill=white]
\tikzstyle{smaptrans}=[draw,shape=sSWbox,fill=white]
\tikzstyle{smapconj}=[draw,shape=sNWbox,fill=white]

\tikzstyle{dsmap}=[draw,dashed,shape=sNEbox,fill=white]
\tikzstyle{dsmapdag}=[draw,dashed,shape=sSEbox,fill=white]
\tikzstyle{dsmaptrans}=[draw,dashed,shape=sSWbox,fill=white]
\tikzstyle{dsmapconj}=[draw,dashed,shape=sNWbox,fill=white]

\boxshape{mNEbox}{0pt}{10pt}
\boxshape{mSEbox}{0pt}{-10pt}
\boxshape{mNWbox}{10pt}{0pt}
\boxshape{mSWbox}{-10pt}{0pt}
\tikzstyle{mmap}=[draw,shape=mNEbox]
\tikzstyle{mmapdag}=[draw,shape=mSEbox]
\tikzstyle{mmaptrans}=[draw,shape=mSWbox]
\tikzstyle{mmapconj}=[draw,shape=mNWbox]

\tikzstyle{mmapgray}=[draw,fill=gray!40!white,shape=mNEbox]
\tikzstyle{smapgray}=[draw,fill=gray!40!white,shape=sNEbox]

\makeatletter

\pgfdeclareshape{cornerpoint}{
\inheritsavedanchors[from=rectangle] 
\inheritanchorborder[from=rectangle]
\inheritanchor[from=rectangle]{center}
\inheritanchor[from=rectangle]{north}
\inheritanchor[from=rectangle]{south}
\inheritanchor[from=rectangle]{west}
\inheritanchor[from=rectangle]{east}
\backgroundpath{
\southwest \pgf@xa=\pgf@x \pgf@ya=\pgf@y
\northeast \pgf@xb=\pgf@x \pgf@yb=\pgf@y

\pgfmathsetmacro{\pgf@shorten@left}{\pgfkeysvalueof{/tikz/shorten left}}
\pgfmathsetmacro{\pgf@shorten@right}{\pgfkeysvalueof{/tikz/shorten right}}

\pgfpathmoveto{\pgfpoint{0.5 * (\pgf@xa + \pgf@xb)}{\pgf@ya - 5pt}}
\pgfpathlineto{\pgfpoint{\pgf@xa - 8pt + \pgf@shorten@left}{\pgf@yb - 1.5 * \pgf@shorten@left}}
\pgfpathlineto{\pgfpoint{\pgf@xa - 8pt + \pgf@shorten@left}{\pgf@yb}}
\pgfpathlineto{\pgfpoint{\pgf@xb + 8pt - \pgf@shorten@right}{\pgf@yb}}
\pgfpathlineto{\pgfpoint{\pgf@xb + 8pt - \pgf@shorten@right}{\pgf@yb - 1.5 * \pgf@shorten@right}}
\pgfpathclose
}
}

\pgfdeclareshape{cornercopoint}{
\inheritsavedanchors[from=rectangle] 
\inheritanchorborder[from=rectangle]
\inheritanchor[from=rectangle]{center}
\inheritanchor[from=rectangle]{north}
\inheritanchor[from=rectangle]{south}
\inheritanchor[from=rectangle]{west}
\inheritanchor[from=rectangle]{east}
\backgroundpath{
\southwest \pgf@xa=\pgf@x \pgf@ya=\pgf@y
\northeast \pgf@xb=\pgf@x \pgf@yb=\pgf@y

\pgfmathsetmacro{\pgf@shorten@left}{\pgfkeysvalueof{/tikz/shorten left}}
\pgfmathsetmacro{\pgf@shorten@right}{\pgfkeysvalueof{/tikz/shorten right}}

\pgfpathmoveto{\pgfpoint{0.5 * (\pgf@xa + \pgf@xb)}{\pgf@yb + 5pt}}
\pgfpathlineto{\pgfpoint{\pgf@xa - 8pt + \pgf@shorten@left}{\pgf@ya + 1.5 * \pgf@shorten@left}}
\pgfpathlineto{\pgfpoint{\pgf@xa - 8pt + \pgf@shorten@left}{\pgf@ya}}
\pgfpathlineto{\pgfpoint{\pgf@xb + 8pt - \pgf@shorten@right}{\pgf@ya}}
\pgfpathlineto{\pgfpoint{\pgf@xb + 8pt - \pgf@shorten@right}{\pgf@ya + 1.5 * \pgf@shorten@right}}
\pgfpathclose
}
}

\makeatother

\pgfkeyssetvalue{/tikz/shorten left}{0pt}
\pgfkeyssetvalue{/tikz/shorten right}{0pt}

\tikzstyle{kpoint common}=[draw,fill=white,inner sep=1pt,minimum height=4mm]
\tikzstyle{kpoint sc}=[shape=cornerpoint,kpoint common]
\tikzstyle{kpoint adjoint sc}=[shape=cornercopoint,kpoint common]
\tikzstyle{kpoint}=[shape=cornerpoint,shorten left=5pt,kpoint common]
\tikzstyle{kpoint adjoint}=[shape=cornercopoint,shorten left=5pt,kpoint common]
\tikzstyle{kpoint conjugate}=[shape=cornerpoint,shorten right=5pt,kpoint common]
\tikzstyle{kpoint transpose}=[shape=cornercopoint,shorten right=5pt,kpoint common]
\tikzstyle{kpoint symm}=[shape=cornerpoint,shorten left=5pt,shorten right=5pt,kpoint common]

\tikzstyle{wide kpoint sc}=[shape=cornerpoint,kpoint common, minimum width=1 cm]
\tikzstyle{wide kpointdag sc}=[shape=cornercopoint,kpoint common, minimum width=1 cm]

\tikzstyle{black kpoint}=[shape=cornerpoint,shorten left=5pt,kpoint common,fill=black,font=\color{white}]

\tikzstyle{black kpoint sm}=[shape=cornerpoint,shorten left=5pt,kpoint common,fill=black,font=\color{white},scale=0.75]

\tikzstyle{black kpoint adjoint}=[shape=cornercopoint,shorten left=5pt,kpoint common,fill=black,font=\color{white}]
\tikzstyle{black kpointadj}=[shape=cornercopoint,shorten left=5pt,kpoint common,fill=black,font=\color{white}]

\tikzstyle{black kpointadj sm}=[shape=cornercopoint,shorten left=5pt,kpoint common,fill=black,font=\color{white},scale=0.75]

\tikzstyle{black dkpoint}=[shape=cornerpoint,shorten left=5pt,kpoint common,fill=black, doubled,font=\color{white}]
\tikzstyle{black dkpoint adjoint}=[shape=cornercopoint,shorten left=5pt,kpoint common,fill=black, doubled,font=\color{white}]
\tikzstyle{black dkpointadj}=[shape=cornercopoint,shorten left=5pt,kpoint common,fill=black, doubled,font=\color{white}]

\tikzstyle{black dkpoint sm}=[shape=cornerpoint,shorten left=5pt,kpoint common,fill=black, doubled,font=\color{white},scale=0.75]
\tikzstyle{black dkpointadj sm}=[shape=cornercopoint,shorten left=5pt,kpoint common,fill=black, doubled,font=\color{white},scale=0.75]

\tikzstyle{kpointdag}=[kpoint adjoint]
\tikzstyle{kpointadj}=[kpoint adjoint]
\tikzstyle{kpointconj}=[kpoint conjugate]
\tikzstyle{kpointtrans}=[kpoint transpose]

\tikzstyle{big kpoint}=[kpoint, minimum width=1.2 cm, minimum height=8mm, inner sep=4pt, text depth=3mm]

\tikzstyle{wide kpoint}=[kpoint, minimum width=1 cm, inner sep=2pt]
\tikzstyle{wide kpointdag}=[kpointdag, minimum width=1 cm, inner sep=2pt]
\tikzstyle{wide kpointconj}=[kpointconj, minimum width=1 cm, inner sep=2pt]
\tikzstyle{wide kpointtrans}=[kpointtrans, minimum width=1 cm, inner sep=2pt]

\tikzstyle{wider kpoint}=[kpoint, minimum width=1.25 cm, inner sep=2pt]
\tikzstyle{wider kpointdag}=[kpointdag, minimum width=1.25 cm, inner sep=2pt]
\tikzstyle{wider kpointconj}=[kpointconj, minimum width=1.25 cm, inner sep=2pt]
\tikzstyle{wider kpointtrans}=[kpointtrans, minimum width=1.25 cm, inner sep=2pt]

\tikzstyle{gray kpoint}=[kpoint,fill=gray!50!white]
\tikzstyle{gray kpointdag}=[kpointdag,fill=gray!50!white]
\tikzstyle{gray kpointadj}=[kpointadj,fill=gray!50!white]
\tikzstyle{gray kpointconj}=[kpointconj,fill=gray!50!white]
\tikzstyle{gray kpointtrans}=[kpointtrans,fill=gray!50!white]

\tikzstyle{gray dkpoint}=[kpoint,fill=gray!50!white,doubled]
\tikzstyle{gray dkpointdag}=[kpointdag,fill=gray!50!white,doubled]
\tikzstyle{gray dkpointadj}=[kpointadj,fill=gray!50!white,doubled]
\tikzstyle{gray dkpointconj}=[kpointconj,fill=gray!50!white,doubled]
\tikzstyle{gray dkpointtrans}=[kpointtrans,fill=gray!50!white,doubled]

\tikzstyle{white label}=[draw,fill=white,rectangle,inner sep=0.7 mm]
\tikzstyle{gray label}=[draw,fill=gray!50!white,rectangle,inner sep=0.7 mm]
\tikzstyle{black label}=[draw,fill=black,rectangle,inner sep=0.7 mm]

\tikzstyle{dkpoint}=[kpoint,doubled]
\tikzstyle{wide dkpoint}=[wide kpoint,doubled]
\tikzstyle{dkpointdag}=[kpoint adjoint,doubled]
\tikzstyle{wide dkpointdag}=[wide kpointdag,doubled]
\tikzstyle{dkcopoint}=[kpoint adjoint,doubled]
\tikzstyle{dkpointadj}=[kpoint adjoint,doubled]
\tikzstyle{dkpointconj}=[kpoint conjugate,doubled]
\tikzstyle{dkpointtrans}=[kpoint transpose,doubled]

\tikzstyle{kscalar}=[kpoint common, shape=EBox, inner xsep=-1pt, inner ysep=3pt,font=\small]
\tikzstyle{kscalarconj}=[kpoint common, shape=WBox, inner xsep=-1pt, inner ysep=3pt,font=\small]

\tikzstyle{spekpoint}=[kpoint sc,minimum height=5mm,inner sep=3pt]
\tikzstyle{spekcopoint}=[kpoint adjoint sc,minimum height=5mm,inner sep=3pt]

\tikzstyle{dspekpoint}=[spekpoint,doubled]
\tikzstyle{dspekcopoint}=[spekcopoint,doubled]

 \tikzstyle{upground}=[circuit ee IEC,thick,ground,rotate=90,scale=2.5]
 \tikzstyle{downground}=[circuit ee IEC,thick,ground,rotate=-90,scale=2.5]
 \tikzstyle{bigground}=[regular polygon,regular polygon sides=3,draw=gray,scale=0.50,inner sep=-0.5pt,minimum width=10mm,fill=gray]

\tikzstyle{arrs}=[-latex,font=\small,auto]
\tikzstyle{arrow plain}=[arrs]
\tikzstyle{arrow dashed}=[dashed,arrs]
\tikzstyle{arrow bold}=[very thick,arrs]
\tikzstyle{arrow hide}=[draw=white!0,-]
\tikzstyle{arrow reverse}=[latex-]
\tikzstyle{cdnode}=[]

\tikzstyle{discarding}=[fill=white, draw=black, shape=circle, style=upground]
\tikzstyle{smalldiscarding}=[fill=white, draw=black, style=upground, scale=0.5]
\tikzstyle{backdiscard}=[fill=white, draw=black, shape=circle, style=downground, scale=0.5]
\tikzstyle{smallbackdiscard}=[fill=white, draw=black, shape=circle, style=downground, scale=0.5]
\tikzstyle{state}=[fill=white, draw=black, style=triang, tikzit shape=rectangle]
\tikzstyle{kstate}=[fill=white, draw=black, style=kpoint, tikzit shape=rectangle]
\tikzstyle{kstateconj}=[fill=white, draw=black, style=kpoint conjugate, tikzit shape=rectangle]
\tikzstyle{kstateBIG}=[fill=white, draw=black, style=big kpoint, tikzit shape=rectangle]
\tikzstyle{effect}=[fill=white, draw=black, style=triangdag]
\tikzstyle{keffect}=[fill=white, draw=black, style=kpoint adjoint]
\tikzstyle{keffectconj}=[fill=white, draw=black, style=kpoint transpose]
\tikzstyle{morphdag}=[style=mapdag]
\tikzstyle{morph}=[style=hadamard]
\tikzstyle{WIDEmorph}=[style=hadamard, minimum width=14mm]
\tikzstyle{morphtrans}=[style=maptrans]
\tikzstyle{morphconj}=[style=mapconj]
\tikzstyle{CPMmorph}=[style=dmap]
\tikzstyle{CPMmorphconj}=[style=dmapconj]
\tikzstyle{CPMmorphdag}=[style=dmapdag]
\tikzstyle{CPMmorphtrans}=[style=dmaptrans]
\tikzstyle{CPMstate}=[fill=white, draw=black, style=triang, doubled]
\tikzstyle{CPMstateBIG}=[fill=white, draw=black, style={triang_lesssep}, doubled]
\tikzstyle{CPMkstate}=[fill=white, draw=black, style=kpoint, tikzit shape=rectangle, doubled]
\tikzstyle{CPMkstateconj}=[fill=white, draw=black, style=kpoint conjugate, tikzit shape=rectangle, doubled]
\tikzstyle{CPMkstateBIG}=[fill=white, draw=black, style=big kpoint, tikzit shape=rectangle, doubled]
\tikzstyle{CPMkeffect}=[fill=white, draw=black, style=kpoint adjoint, doubled]
\tikzstyle{CPMkeffectconj}=[fill=white, draw=black, style=kpoint transpose, doubled]
\tikzstyle{UHfB}=[fill=white, draw=black, style=triangdag, doubled, inner sep=-2pt]
\tikzstyle{leak}=[style=tinypoint, regular polygon rotate=-90]
\tikzstyle{leakfill}=[style=tinypoint, regular polygon rotate=-90, fill=black]
\tikzstyle{Z}=[style=dot, fill=green]
\tikzstyle{X}=[style=dot, fill=red]
\tikzstyle{black_dot}=[style=dot, fill=black]
\tikzstyle{white_dot}=[style=dot, fill=white]
\tikzstyle{qblack_dot}=[style=ddot, fill=black]
\tikzstyle{qwhite_dot}=[style=ddot, fill=white]
\tikzstyle{whitephase}=[style=wphase dot, fill=white]
\tikzstyle{qredphase}=[style=phase dot, fill=red]
\tikzstyle{qgreenphase}=[style=phase dot, fill=green]
\tikzstyle{had}=[style=hadamard, doubled]
\tikzstyle{box}=[style=hadamard]
\tikzstyle{classhad}=[style=hadamard]
\tikzstyle{antipode}=[style=anti]

\tikzstyle{dottededge}=[-, dotted]
\tikzstyle{double edge}=[-, style=doubled, draw=black, tikzit draw={rgb,255: red,234; green,209; blue,255}]
\tikzstyle{new edge style 0}=[<-]
\tikzstyle{new edge style 1}=[-, draw={rgb,255: red,234; green,209; blue,255}, fill={rgb,255: red,234; green,209; blue,255}]
\tikzstyle{new edge style 1b}=[-, draw={rgb,255: red,208; green,218; blue,216}, fill={rgb,255: red,208; green,218; blue,216}]
\tikzstyle{new edge style 2}=[-, draw={rgb,255: red,14; green,188; blue,83}]
\tikzstyle{new edge style 3}=[<-, draw={rgb,255: red,234; green,209; blue,255}]
\tikzstyle{new edge style 4}=[<-, draw={rgb,255: red,0; green,128; blue,128}]
\tikzstyle{new edge style 5}=[-, draw={rgb,255: red,214; green,110; blue,62}]
\tikzstyle{new edge style 6}=[-, draw={rgb,255: red,174; green,20; blue,174}]

\usetikzlibrary{calc, positioning, shapes.geometric}
\usetikzlibrary{
	arrows,
	shapes,
	decorations,
	intersections,
	backgrounds,
	positioning,
	circuits.ee.IEC
	}

\newcommand{\tikzfigscale}[2]{\scalebox{#1}{\input{#2.tikz}}}

\def\be{\begin{equation}}
\def\ee{\end{equation}}
\def\ba{\begin{align}}
\def\ea{\end{align}}

\newcommand{\ce}{\mathcal E}

\newcommand{\ci}{\mathcal I}

\usepackage{bbm}

\newcommand{\lalg}[1]{\mathfrak{#1}}  

\renewcommand{\sl}{\lalg{sl}}

\newtheorem{definition}{Definition}
\newtheorem{theorem}{Theorem}
\newtheorem*{theorem*}{Theorem}

\newtheorem{lemma}{Lemma}

\newtheorem{example}{Example}

\tikzstyle{every picture}=[baseline=-0.25em,shorten <=-0.1pt]
\tikzstyle{dotpic}=[scale=0.5]
\tikzstyle{braceedge}=[decorate,decoration={brace,amplitude=1mm,raise=-1mm}]
\tikzstyle{dot}=[inner sep=0.7mm,minimum width=0pt,minimum height=0pt,fill=black,draw=black,shape=circle]
\tikzstyle{small dot}=[inner sep=0.1mm,minimum width=0pt,minimum height=0pt,fill=black,draw=black,shape=circle]
\tikzstyle{black dot}=[dot]
\tikzstyle{white dot}=[dot,fill=white]
\tikzstyle{gray dot}=[dot,fill=gray!40!white]
\tikzstyle{alt white dot}=[white dot,label={[xshift=3mm,yshift=-0.05mm,font=\tiny]left:$*$}]
\tikzstyle{alt gray dot}=[gray dot,label={[xshift=3mm,yshift=-0.05mm,font=\tiny]left:$*$}]
\tikzstyle{white norm}=[rectangle,fill=white,draw=black,minimum height=2mm,minimum width=2mm,inner sep=0pt,font=\small]
\tikzstyle{gray norm}=[white norm,fill=gray!40!white]
\tikzstyle{square box}=[rectangle,fill=white,draw=black,minimum height=5mm,minimum width=5mm,font=\small]
\tikzstyle{square gray box}=[rectangle,fill=gray!30,draw=black,minimum height=6mm,minimum width=6mm]
\tikzstyle{diredge}=[->]
\tikzstyle{rdiredge}=[<-]
\tikzstyle{dashed edge}=[dashed]
\tikzstyle{cross}=[preaction={draw=white, -, line width=3pt}]

\newcommand{\dotdualmult}[1]{%
\!\begin{tikzpicture}[dotpic]
    \node [style=white dot] (0) at (0, 0.3) {};
    \node [style=none] (1) at (-0.5, -0.4) {};
    \node [style=none] (2) at (0.5, -0.4) {};
    \node [style=none] (3) at (0, 0.8) {};
    \draw [style=diredge] (3.center) to (0);
    \draw [style=diredge, in=15, out=-30, looseness=1.50] (0) to (1.center);
    \draw [style=diredge, in=165, out=-150, looseness=1.50] (0) to (2.center);
\end{tikzpicture}\!}

\newcommand{\dotconorm}[1]{%
\,\begin{tikzpicture}[dotpic,yshift=0.4mm]
    \node [style=none] (0) at (0, -0.4) {};
    \node [style=white norm] (1) at (0, 0.1) {};
    \node [style=none] (2) at (0, 0.5) {};
    \draw [style=diredge] (1) to (0.center);
    \draw (2.center) to (1);
\end{tikzpicture}\,}

\DeclarePairedDelimiterX\ketbra[2]{\lvert}{\rvert}{#1\delimsize\rangle\!\delimsize\langle#2}
\DeclarePairedDelimiterX\projector[1]{\lvert}{\rvert}{#1\delimsize\rangle\!\delimsize\langle#1}%

\newcommand{\astfootnote}[1]{
\let\oldthefootnote=\thefootnote
\setcounter{footnote}{0}
\renewcommand{\thefootnote}{\fnsymbol{footnote}}
\footnote{#1}
\let\thefootnote=\oldthefootnote
}

\title{On the Origin of Linearity and Unitarity in Quantum Theory}

\author{Matt Wilson}
\affiliation{Universit\'{e} Paris-Saclay, CNRS, ENS Paris-Saclay, Inria, CentraleSup\'{e}lec, Laboratoire M\'{e}thodes Formelles}
\affiliation{PPLV Group, Department of Computer Science, University College London}
\email{matthew.wilson@centralesupelec.fr}
\author{Nick Ormrod}
\affiliation{Quantum Group, Computer Science Department, University of Oxford}
\email{nicholas.ormrod@cs.ox.ac.uk}

\usepackage{hyperref}
\makeatletter
\def\HyPsd@expand@utfvii{}
\makeatother

\begin{document} \emergencystretch 3em

\maketitle

\begin{abstract}
We reconstruct the transformations of quantum theory using a physically motivated postulate. This postulate states that transformations should be \textit{locally applicable}, and recovers the linear isometries from pure quantum theory, as well as the completely positive, trace-preserving maps from mixed quantum theory. Notably, in the pure case, linearity with respect to the superposition rule and reversibility are both derived from this locality principle. 
\end{abstract}






\section{Introduction}
Why do pure quantum states evolve linearly and unitarily, and why do mixed states evolve according to completely positive trace-preserving maps? Here is a way to make this question more precise. Suppose one has a theory that describes the states of physical systems and measurements on them in the same way that quantum theory does. Is it necessary that this theory should treat transformations in the same way too? Or is there some reasonable modification of quantum theory that changes the transformations while leaving the rest intact \cite{WEINBERG1989336_testing_qm, first_schroddy_newter, Stamp_2012_correlated_worldline_1}?

Of course, the answer depends on what is meant by `reasonable'.
Here is a suggestion: a reasonable transformation on a system is one that is \textit{locally applicable}. By this we mean, roughly speaking, that one can imagine that the system is accompanied by some (possibly far away) environment, on which the transformation does not act. A hint that such a principle could be used to recover standard quantum transformations can be found in a pair of recent papers \cite{wilson2023quantum, wilson2022free}, in which a formalization of local applicability in a rather different context was used to re-characterise the quantum supermaps \cite{Chiribella2008TransformingSupermaps, Chiribella_2013_without_causal}, also known as process matrices \cite{Oreshkov_2012_process_matrices}.

\begin{figure}[H]
    \centering
    \begin{tikzpicture}
	\begin{pgfonlayer}{nodelayer}
		\node [style=none] (150) at (0, 2) {};
		\node [style=none] (151) at (-3.5, 2) {};
		\node [style=none] (152) at (-3.5, -2) {};
		\node [style=none] (153) at (0, -2) {};
		\node [style=none] (154) at (3.5, 2) {};
		\node [style=none] (155) at (0, 2) {};
		\node [style=none] (156) at (0, -2) {};
		\node [style=none] (157) at (3.5, -2) {};
		\node [style=none] (158) at (-1.75, -2.5) {$\texttt{System}$};
		\node [style=none] (159) at (1.75, -2.5) {$\texttt{Environment}$};
		\node [style=box] (160) at (-1.75, 0) {$\mathscr{L}$};
		\node [style={qwhite_dot}] (161) at (1.75, 0) {$ \ \mathbb{m} \ $};
		\node [style=none] (163) at (1.75, -1) {};
		\node [style=none] (164) at (-1.75, -1) {};
		\node [style=none] (165) at (-1.75, 1) {};
		\node [style=none] (166) at (1.75, 1) {};
		\node [style=none] (167) at (2.25, -0.5) {};
		\node [style=none] (168) at (-1.75, -2) {};
		\node [style=none] (169) at (1.75, -2) {};
		\node [style=none] (170) at (-1.75, 2) {};
		\node [style=none] (171) at (1.75, 2) {};
	\end{pgfonlayer}
	\begin{pgfonlayer}{edgelayer}
		\draw [fill={blue!10}, draw={blue!20}] (153.center)
			 to (150.center)
			 to (151.center)
			 to (152.center)
			 to cycle;
		\draw [fill={red!10}, draw={blue!20}] (157.center)
			 to (154.center)
			 to (155.center)
			 to (156.center)
			 to cycle
			 to cycle;
		\draw (160) to (164.center);
		\draw (161) to (163.center);
		\draw [style=double edge] (161) to (167.center);
		\draw [style=new edge style 0] (164.center) to (168.center);
		\draw [style=new edge style 0] (163.center) to (169.center);
		\draw [style=new edge style 0] (165.center) to (160);
		\draw [style=new edge style 0] (166.center) to (161);
		\draw (170.center) to (165.center);
		\draw (171.center) to (166.center);
	\end{pgfonlayer}
\end{tikzpicture} 
    \caption{The intuition behind local applicability. The system is separate from its environment; the locally applicable transformation $\mathscr{L}$ acts only on the system; and in particular it is independent of measurements performed on the environment\protect\footnotemark.}
    \label{fig:my_label}
\end{figure}
\footnotetext{One of the authors was proud to have managed to represent a measurement using a picture of a magnifying glass in TikZ; the other simply assumed that there was a glitch.}

In this paper, we provide an answer to this question, restricting ourselves to the finite-dimensional case for simplicity. Given the standard treatment of quantum states and measurements, the quantum dynamics can be derived from the postulate of local applicability alone. Conversely, given the standard treatment of quantum states and measurement, all modified versions of the quantum dynamics are not locally applicable. 

We begin by introducing the notion of a \textit{state-measurement theory} (Section \ref{sec:state-measurement}). This is a theory that describes states and measurements but not dynamics. We then show that given a \textit{quantum} state-measurement theory, the standard quantum dynamics are the only ones compatible with local applicability (Section \ref{sec:reconstruction}). We do this separately for the ``pure'' and ``mixed'' quantum state-measurement theories. The locally applicable transformations on the pure quantum state-measurement theory are precisely the unitary ones (or isometric in cases where input and output dimensions do not match). The locally applicable transformations on the mixed quantum state-measurement theory are precisely the quantum channels. 

After giving our main results, we compare them to some due to Gisin and later collaborators \cite{Gisin:1989sx_OG_nosig_paper, Simon_2001_gisin_CPTP} (Section \ref{sec:gisin}). In an influential paper \cite{Gisin:1989sx_OG_nosig_paper}, Gisin outlines an argument for the Schrödinger evolution from determinism, understood as the requirement that pure states are mapped to pure states, and compatibility with relativity, understood as a prohibition on superluminal signaling. From the perspective of this paper, the approach of \cite{Gisin:1989sx_OG_nosig_paper} has a few shortcomings. Firstly, it does not recover linearity with respect to the superposition rule for discrete time evolution. Secondly, it requires an additional assumption of complete positivity in the mixed setting. Thirdly, it is unclear how the approach in \cite{Gisin:1989sx_OG_nosig_paper} could be generalized to derive linearity for pure functions between distinct Hilbert spaces of possibly different dimensions. The derivation in this paper differs on all three counts.


The implications of our main results separate into two main categories. On the one hand, our results impose constraints on modifications to quantum theory. Specifically, they constitute no-go theorems for non-linear, non-unitary, or non-completely positive modifications to the quantum dynamics. This is relevant for proposed modifications that reject unitarity in order to avoid a measurement problem \cite{ghirardi1986unified, pearle1989combining, ghirardi1990markov, diosi1987universal, diosi1989models, penrose1996gravity, penrose2014gravitization, ghirardi1990relativistic, tumulka2006relativistic, ormrod2023theories}, and that aim to model interactions between quantum particles and classical gravitational fields \cite{WEINBERG1989336_testing_qm, weinberg_precision_tests, first_schroddy_newter, Carlip_2008_non_qm_gravity, rembieli_quasi_convex, ray2023nosignaling, testing_quantum_gravity, Stamp_2012_correlated_worldline_1,Stamp_2015_argument_correlated_worldline, barvinsky_structure_worldline, barvinsky_correlated_worldline_2, gerow_propogator_correlated_worldline, oppenheim2021postquantum}. Some key insights of this particular no-go theorem from local-applicability are that neither varying system dimension over time nor considering discrete time-evolution will help to evade linearity. With regards to continuous global time-parameters, we wonder if the results might therefore be relevent in the context of quantum gravity, where the the lack of a global time parameter and possible discreteness of spacetime arise \cite{isham1992, rovelli2014covariant, CHRISTODOULOU201964, Christodoulou2022experimenttotest, Surya_2019}. 

On the other hand, our results shed light on the structure of quantum theory itself, and its connection to the other major pillar of modern physics. Local applicability is motivated by relativity theory (including by the requirement in relativistic field theories that spacelike separated operators commute, and by the prohibition on superluminal signalling). Our results therefore tell us that if the states and measurements of quantum theory can be taken for granted, then the full quantum theory can be derived using a single relativistically-motivated physical postulate. This contributes to the longstanding research effort of reconstructing quantum theory, whether from diagrammatic \cite{Selby_2021_diagram}, relativistic \cite{Gisin:1989sx_OG_nosig_paper, Simon_2001_gisin_CPTP, Gisin1990WeinbergsNQ, Polchinski_signals, JORDAN1999263_signals},  operational \cite{Chiribella_2010_prob_pure, hardy2013reconstructing, Hardy2001QuantumAxioms, Masanes_2019, Galley2018anymodificationof, Galley2017classificationofall, Torre_PhysRevLett.109.090403, masanes_derivation, Chiribella_2018}, categorical \cite{Huot_2019_universal, Coecke2008AxiomaticCPM-construction, heunen_axioms, tull2019categorical, Wetering_2019, gogioso2019processtheoretic}, or computational \cite{Popescu1994-POPQNA,Westerbaan_2022, Bao_2016}  considerations.

\section{State-measurement Theories} \label{sec:state-measurement}

We begin by introducing the concept of a state-measurement theory.


\subsection{Basic Structure}

Take away the unitary dynamics from ``pure'' quantum theory, and what is one left with? One still has the idea that systems are represented with Hilbert spaces. One still has the idea that states are represented by the normalized vectors. And one still has the idea that measurement outcomes are represented by projectors. In addition, one has the Born rule for the probabilities of measurement outcomes, and the projection postulate for calculating the state the system is left in after it is measured. Altogether, this constitutes an example of what we will call a \textit{state-measurement} (SM) \textit{theory}.

More generally, a state-measurement theory does two things. Given any state and possible measurment outcome, it determines (1) the probability of getting that outcome, and (2) the state of the system after that outcome is obtained. That is, it provides both a probability rule and a state-update rule. These ideas are illustrated in Figure \ref{fig:state_measurement}, and formalized below.

\begin{figure} 
    \centering
    \[   \tikzfigscale{0.8}{figs/basic_update_b}  \] 
    \caption{The basic idea of a state-measurement theory. The theory allows one to calculate the probability of obtaining a measurement outcome $\mathbb{m}$ given the initial state $s$, as well as the updated state $s'$ after the measurement.
    \label{fig:state_measurement}}
\end{figure}

\begin{definition} \label{def:state_measurement_theory}
A state-measurement theory consists of a set of system $O=\{A, B, \ldots \}$, and for each system $A$
\begin{itemize}
    \item a set of states $S_A$;
    \item a set of (representatives of) measurement outcomes $M_A$;
    \item a probability function $p_A:S_A \times M_A \rightarrow [0,1]$;
    \item a partial function $u_A: S_A \times M_A \rightarrow S_A$ expressing the updated description of a state after a measurement outcome is obtained; and finally
    \item a `null' measurement $\mathbb{I}_A$ satisfying $u_A(s,\mathbb{I}_A) = s$ and $p_A(s,\mathbb{I}_A) = 1$ for all $s \in S_A$.
\end{itemize}
\end{definition} 

Let us now formalize our motivating example of a state-measurement theory, pure quantum SM theory. Throughout the paper, we restrict for simplicity to the case of finite-dimensional quantum theory.

\begin{example}[Pure Quantum SM Theory]
In pure quantum SM theory each system $A$ is a finite-dimensional Hilbert space. The states of $A$ are the vectors in $A$ with norm $1$. The set $M_A$ of measurements on $A$ is given by the set $\Pi(A)$ of orthogonal projectors on $A$. The probability rule is given by \[p(\ket{\psi},\pi) := \bra{\psi} \pi \ket{\psi},\] the update is given by \[    u(\ket{\psi}, \pi) := \frac{\pi \ket{\psi}}{\sqrt{p(\ket{\psi}, \pi)}},    \] and the nothing measurement is given by the identity map which is indeed an orthogonal projector. 
\end{example}


With this example in hand, we make two clarifications regarding Definition \ref{def:state_measurement_theory}. Firstly, outcomes in $M_A$ can be thought of as independent of the measurement context. For example, in pure quantum SM theory, $M_A$ just includes a single copy of each projector, rather than a different copy for each projector-valued measurement in which it is included. Secondly, $u_A$ is only required to be a partial function because for a given state in $S_A$, not all outcomes are possible. For example, if the system $A$ is prepared in the state $\ket{0}_A$, the $\ket{1}\bra{1}_A$ outcome is never obtained, and so pure quantum SM theory is not obligated to prescribe the state of the system following that outcome.


\subsection{Spatial Composition}


A general state-measurement theory needn't have any notion of \textit{joint} systems comprised of multiple subsystems. In particular, there might not be any way of considering combinations of states or measurements that are performed at the same time (see Figure \ref{fig:joint}). Since such notions are a prerequisite for discussions of locality, let us now define the special class of state-measurement theories that have them.

\begin{figure}
    \centering
    $\tikzfigscale{0.8}{figs/space_state_2} \ \  \ \ \ \ \ \ \ \ \ \ \ \ \ \tikzfigscale{0.8}{figs/space_update_a}$
    \caption{The intuition behind spatial state-measurement theories. Particularly if systems are objects in space, one should be able to consider combinations of those systems along with their states and their measurements. The former diagram expresses the notion that any two states $s$ and $r$ can be combined via tensor product to produce a joint state (on which any measurement outcome of the joint system can be observed). The latter diagram expresses the notion that any two measurement outcomes $\mathbb{m}$ and $\mathbb{n}$ can be combined in parallel to produce a joint measurement outcome (which can be observed on states of the joint system).
    }
    \label{fig:joint}
\end{figure}

\begin{definition}
    A spatial state-measurement theory is a state-measurement theory equipped with
    \begin{itemize}
        \item An associative function $\otimes: O \times O \rightarrow O$, meaning $(A \otimes B) \otimes C = A \otimes (B \otimes C)$
        \item For each $A,B$ an associative family of functions $\otimes_{AB}: S_A \times S_B \rightarrow S_{A \otimes B}$, meaning that $(s \otimes r) \otimes t = s \otimes (r \otimes t)$
        \item For each $A,B$ an associative family of functions $\otimes: M_A \times M_B \rightarrow M_{A \otimes B}$, meaning that $(\mathbb{m} \otimes \mathbb{n}) \otimes \mathbb{o} = \mathbb{m} \otimes (\mathbb{n} \otimes \mathbb{o})$
    \end{itemize}
\end{definition}

When it aids readability, we will adopt the convention of writing tensor products of systems without the symbol $\otimes$, for instance denoting $A \otimes B$ simply by $AB$.

Pure quantum SM theory is an example of a spatial state-measurement theory, where the composition function $\otimes$ is the usual tensor product.


\subsection{Locally Applicable Transformations}

We are finally in a position to define local applicability. A transformation on the system $A$ should certainly give rise to a function on the states of $A$, if the transformation is locally applicable, one should also be able to consider its action on $A \otimes X$, where $X$ is some distant environment, and find that it `only acts on $A$'. This idea is formalized below.

\begin{definition} \label{def:local_applicability}
Let $A,B$ be systems of a spatial state-measurement theory. A locally applicable transformation of type $\mathscr{L}:A \rightarrow B$ is a function $\mathscr{L}_{[-]}:S_A \rightarrow S_B$ and a family $\mathscr{L}_X : S_{AX} \rightarrow S_{BX}$ of functions such that
\begin{itemize}
\item State locality. Parallel composition commutes with the action of $\mathscr{L}$. {Formally, for all $X, X' \in O$, $s \in S_{AX}$ and $r \in S_{X'}$, we have $\mathscr{L}_{XX'}(s \otimes r) = \mathscr{L}_{X}(s) \otimes r$.}

    \item No signaling. For every state $s \in S_{AX}$ and measurement $ \mathbb{m} \in M_X$, the probability of $\mathbb{m}$ is not affected by $\mathscr{L}$. Formally, $p_{BX}(\mathscr{L}_{X}(s) , \mathbb{I} \otimes \mathbb{m} )= p_{AX}(s , \mathbb{I} \otimes \mathbb{m} )$.
    \item Update commutativity. It does not matter whether one first acts locally on $A$ and then obtains a measurement outcome on the environment, or the other way round.  Formally, for every state $s \in S_{AX}$ and measurement outcome $\mathbb{m} \in M_X$, $u_{BX}(\mathscr{L}_{X}(s) , \mathbb{I} \otimes \mathbb{m} )= \mathscr{L}_{X}( u_{AX}(s , \mathbb{I} \otimes \mathbb{m} ))$. 
\end{itemize}
\end{definition}

\begin{figure}
    \centering
    $ p \  \tikzfigscale{0.8}{figs/locality_pic_1_b} \ = \ \ \   p \ \tikzfigscale{0.8}{figs/locality_prob_2_b}$
    \caption{Intuition for no-signalling. The probability of some outcome $\mathbb{m}$ after measuring the environment is independent of whether or not the locally applicable transformation $\mathscr{L}$ is performed.}
    \label{fig:my_label}
\end{figure}

\begin{figure}
    \centering
    $ \tikzfigscale{0.8}{figs/locality_update_1_b} \ = \ \ \   \tikzfigscale{0.8}{figs/locality_update_2_b}$
    \caption{Intuition for update commutativity. It does not matter whether one first applies $\mathscr{L}$ and then obtains an outcome or the other way round.}
    \label{fig:my_label}
\end{figure}

Our formalization of a transformation might seem a little cumbersome. On the more familiar way of thinking about transformations in pure quantum theory, a unitary transformation on $A$ can be represented simply by an operator $U$ on its Hilbert space. There is no need to define the transformation as a family functions acting on larger systems, since to apply $U$ to a larger system $A \otimes X$ one can simply take the tensor product with the identity $U \otimes I_X$. Can we not simply define a locally applicable transformation $\mathscr{L}$ as a function on $A$ alone, and assume that the action of $\mathscr{L}$ an a larger system is given by an expression such as $\mathscr{L} \otimes I_X$?

The reason that we cannot do this is that the tensor product is an operation defined on linear maps. Thus the expression $\mathscr{L} \otimes I_X$ presupposes that $\mathscr{L}$ is linear. But we do not want to assume that locally applicable transformations are linear; instead, we want to \textit{derive} that they are. Similarly, we wish to derive rather than assume that $\mathscr{L}_X = \mathscr{L}_{[-]} \otimes I_X$ for all $X \in O$.

\section{Deriving quantum dynamics} \label{sec:reconstruction}

In this section, we derive the familiar quantum dynamics from the postulate of local applicability and the usual treatment of quantum states and measurements.

\subsection{Pure quantum dynamics}

Let us start by explicitly describing the locally applicable transformation on pure quantum theory. It is helpful to introduce the notation $\mathscr{L}_X (\ket{\psi} )^\dag$ for the bra corresponding to the ket $\mathscr{L}_X (\ket{\psi} )$. That is, if $\mathscr{L}_X (\ket{\psi} ) = \ket{\psi'}$, then $\mathscr{L}_X (\ket{\psi} )^\dag = \bra{\psi'}$.

\begin{example}[Locally Applicable Transformations on Pure Quantum SM Theory]
    Let $A,B$ be Hilbert spaces, a locally applicable transformation on pure quantum SM theory of type $A \rightarrow B$ is a function $\mathscr{L}_{[-]}:S_A \rightarrow S_B$ family of functions $\mathscr{L}_X: S_{A X} \rightarrow S_{BX}$ satisfying each of the following.
    \begin{itemize}
    \item State locality:  \[ \mathscr{L}_{XX'}(\ket{\psi} \otimes \ket{\phi} ) = \mathscr{L}_{X}(
\ket{\psi}
    ) \otimes \ket{\phi}, \quad \quad \mathscr{L}_{X}(\ket{\psi}  \otimes \ket{\phi} ) = \mathscr{L}_{[-]}(\ket{\psi} ) \otimes \ket{\phi}  \]
        \item No signaling: \[         \mathscr{L}_X (\ket{\psi} )^\dag (I \otimes \pi) \mathscr{L}_X (\ket{\psi} )  = \bra{ \psi}  I \otimes \pi \ket{\psi } \]
        \item Update commutativity: \[      \frac{(I \otimes \pi) \mathscr{L}_X (\ket{\psi} )}{\sqrt{p(\mathscr{L}_X (\ket{\psi} ), (I \otimes \pi))}} = \mathscr{L}_{X}\Bigg(  \frac{ (I \otimes \pi) \ket{\psi} }{\sqrt{p(\ket{\psi} , I \otimes \pi)}} \Bigg) , \] where we have defined $p(\ket{\psi}  , \pi ):= \bra{ \psi} \pi \ket{\psi } $
    \end{itemize}
\end{example}
Note that using the second bullet point in combination with the third returns \[       \frac{(I \otimes \pi) \mathscr{L}_X (\ket{\psi} )}{\sqrt{p(\ket{\psi} , I \otimes \pi)}} = \mathscr{L}_{X}  \Bigg(\frac{ (I \otimes \pi) \ket{\psi} }{\sqrt{p(\ket{\psi} , I \otimes \pi)}} \Bigg)  .    \] 

Our first result shows that the locally applicable transformations on pure quantum SM theory are precisely the isometries. Of course, it follows that when the input and output systems have the same dimension, they are unitaries.

\begin{theorem} \label{thm:unitary}
    $\mathscr{L}:A \rightarrow B$ is a locally applicable transformation on pure quantum theory if and only if there exists an isometric operator $U: A \rightarrow B$ such that $\mathscr{L}_X = U \otimes I_X$ for all $X \in O$.\footnote{Strictly speaking, this theorem tells us that that $\mathscr{L}_X$ is the function obtained by restricting the domain of $U \otimes I_X$ to normalized vectors. The equality $\mathscr{L}_X = U \otimes I_X$ should be read as stating that $\mathscr{L}_X \ket{\psi}_{AX}= U \otimes I_X \ket{\psi}_{AX}$ for all $\ket{\psi}_{AX} \in S_{AX}$.}
\end{theorem}

To prove this theorem, we rely on the following lemma, which tells us that any isometry defines a locally appplicable transformation.

\begin{lemma}
\label{lemma:unitary}
Let $U:A \rightarrow B$ be an isometric linear map between Hilbert spaces. The family of functions $\mathscr{L}^{U}_X$ defined by \[  \mathscr{L}^{U}_X(\ket{\psi}) :=  (U \otimes I_{X}) \ket{\psi},   \] is a locally-applicable transformation on pure quantum SM theory. 
\end{lemma}

Lemma \ref{lemma:unitary} is proven in Appendix \ref{}. We now prove Theorem \ref{thm:unitary}. 

\begin{proof}
The ``if'' direction follows from Lemma \ref{lemma:unitary}. For the ``only if'' direction, consider an arbitrary system $X \in O$, and an arbitrary state $\ket{\psi}_{AX} \in S_{AX}$. Let $A'$ and $A''$ be duplicates of $A$, and $X'$ and $X''$ be duplicates of $X$. Consider the following teleportation-like experiment.
    \begin{enumerate}
        \item 
        Prepare the six-partite system $A \otimes X\otimes A' \otimes X'\otimes A''\otimes X'' $ in the state  \[\mathscr{L}_{XA'X'}(\ket{\Phi^+}_{AA'} \otimes \ket{\Phi^+}_{XX'}) \otimes \ket{\psi}_{A''X''}\].
        \item Measure $A' \otimes X' \otimes A'' \otimes X''$ in a basis that includes $\ket{\Phi}^+_{A'A''}\ket{\Phi}^+_{X'X''}$.
    \end{enumerate}
    The probability of obtaining the $\ket{\Phi}^+_{A'X'}\ket{\Phi}^+_{A''X''}$ outcome is $\frac{1}{d_A^2d_X^2}$.
    We will show that if this outcome is obtained, then one will be left with the state $\mathscr{L}_{X}(\ket{\psi}_{AX})$ on $A \otimes X$. For this, we need to show that
    \begin{equation} \label{eq:linear}
        \mathscr{L}_{X}(\ket{\psi}_{AX}) = L\ket{\psi}_{A''X''} 
    \end{equation}
    where $L: A'' \otimes X''\rightarrow B \otimes X$ is the linear operator defined by
    \begin{equation}
        L := d_A d_X (I_{AX} \otimes \bra{\Phi^+}_{A'A''} \otimes \bra{\Phi^+}_{X'X''}) \mathscr{L}_{XA'X'}(\ket{\Phi^+}_{AA'} \otimes \ket{\Phi^+}_{XX'}) \otimes I_{A''X''}.
    \end{equation}
 
To begin with, state locality allows us to rewrite
    \begin{equation} 
        \begin{split}
            L \ket{\psi}_{A''X''} &= d_Ad_X (I_{AX} \otimes \bra{\Phi^+}_{A'A''} \otimes \bra{\Phi^+}_{X'X''}) \mathscr{L}_{XA'X'A''X''}(\ket{\Phi^+}_{AA'} \otimes \ket{\Phi^+}_{XX'} \otimes \ket{\psi}_{A''X''}) 
        \end{split}
    \end{equation}
    No signalling implies that the probability of obtaining the $\ket{\Phi}^+_{A'X'}\ket{\Phi}^+_{A''X''}$ outcome from a measurement of $\mathscr{L}_{XA'X'A''X''}(\ket{\Phi^+}_{AA'} \otimes \ket{\Phi^+}_{XX'} \otimes \ket{\psi}_{A''X''})$ is also $\frac{1}{d_A^2d_X^2}$. Update commutativity (and the fact that $\bra{\Phi^+} = \bra{\Phi^+} \ket{\Phi^+} \bra{\Phi^+}$) then implies that 
\begin{equation}
        \begin{split}
            L \ket{\psi}_{A''X''} &= (I_{AX} \otimes \bra{\Phi^+}_{A'A''} \otimes \bra{\Phi^+}_{X'X''})  \\&   \mathscr{L}_{XA'X'A''X''} \Bigg( d_A d_X (I_{AX} \otimes \ket{\Phi^+} \bra{\Phi^+}_{A'A''} \otimes \ket{\Phi^+} \bra{\Phi^+}_{X'X''})\\&   (\ket{\Phi^+}_{AA'} \otimes \ket{\Phi^+}_{XX'} \otimes \ket{\psi}_{A''X''}) \Bigg) \\ 
            &= (I_{AX} \otimes \bra{\Phi^+}_{A'A''} \otimes \bra{\Phi^+}_{X'X''})\mathscr{L}_{XA'X'A''X''}(\ket{\psi}_{AX} \otimes \ket{\Phi^+}_{A'A''} \otimes \ket{\Phi^+}_{X'X''}) 
        \end{split}
    \end{equation}
    A second application of state locality gives us equation (\ref{eq:linear}). 
    The proof of Eq.\ \ref{fig:proof} is summarized by Figure \ref{fig:proof}.
    \begin{figure} 
        \centering 
        \input{figs/big_proof.tikz}
         \caption{The proof of Eq.\ \ref{eq:linear}. Equalities are decorated with the parts of local applicability from which they follow; undecorated equalities follow from basic linear algebra. Wire labels are omitted where it aids readability.}
        \label{fig:proof}
    \end{figure}

    It follows that $\mathscr{L}_X$ itself is a linear operator (or, more precisely, it is the function obtained by restricting the domain of a linear operator to the normalized vectors in the Hilbert space $A \otimes X$.) To show that $\mathscr{L}_X$ is an isometry, it suffices to note that since all locally applicable transformations on pure quantum theory must map normalized states to normalized states
    \begin{equation}
        \mathscr{L}_X (\ket{\psi})^\dag \mathscr{L}_X (\ket{\psi}) = 1
    \end{equation}
    for all $\ket{\psi}$. 
    
    All of the above holds for an arbitrary system $X$. This includes the case where $d_X=1$. Hence $\mathscr{L}_{[-]}$ is an isometry. Finally, state locality implies that for any $X \in O$, we have $\mathscr{L}_X =  \mathscr{L}_{[-]}\otimes I_X$.
\end{proof}

In the context of pure quantum SM theory, locally applicable transformations and isometric operators are in one-one correspondence. We note that this one-one correspondence respects the sequential composition of transformations: the locally applicable transformation $\mathscr{L}^{U_2 \circ U_1}$ corresponding to the isometry $U_2 \circ U_1$ is the one obtained by applying $\mathscr{L}^{U_2}$ after $\mathscr{L}^{U_1}$, i.e.\ $\mathscr{L}^{U_2 \circ U_1}= \mathscr{L}^{U_2} \circ \mathscr{L}^{U_1}$. Similarly, $U^\mathscr{L}= U^{\mathscr{L}_2 \circ \mathscr{L}_1}$. It follows that not only are the locally applicable transformations and isometric operators in one-one correspondence, but they also \textit{define equivalent categories}.

What this shows is that the full pure quantum theory can be derived from (1) its treatment of states and measurements, and (2) local applicability. Put another way, if we assume that states and measurements are treated in the usual way, then the postulate that transformations are locally applicable is equivalent to the postulate that they are isometric.

\subsection{Mixed Quantum Dynamics}

Only closed quantum systems evolve unitarily --- open ones evolve via \textit{quantum channels}. Our next result says that quantum channels are the only locally applicable transformations on a state-measurement theory for open quantum systems. This state-measurement theory is defined as follows.

\begin{example}[Mixed quantum SM theory]
    In mixed quantum SM theory, each system $A$ is a finite-dimensional Hilbert space. The set $S_A$ of states of the system $A$ is the set $D(A)$ of density operators (i.e. the trace-one positive operators) on $A$. The set $M_A$ of measurement outcomes is the set $\Pi(A)$ of projectors on $A$; the nothing outcome is represented by the identity operator. The probability rule $p:D(A) \times \Pi(A) \rightarrow [0,1]$ is given by \[p(\rho, \pi) := {\rm Tr}(\rho \pi),\]
    and the update rule $u:D(A) \times \Pi(A) \rightarrow D(A)$ is given by \[u(\rho, \pi) = \frac{\pi \rho \pi}{p(\rho, \pi)}.\] This is a spatial state-update theory where parallel composition on both states and measurements is given by the tensor product.
\end{example}

More generally, one could include arbitrary positive-operator valued measurements (POVMs) in the state-measurement theory, or, even more generally, arbitrary quantum instruments. In these cases, the following lemma and theorems still hold. Let us now describe the locally applicable transformations on mixed quantum SM theory.

\begin{example}[Locally applicable transformations on static mixed quantum SM theory]
    Let $A,B$ be Hilbert spaces. A locally applicable transformation on mixed quantum SM theory of type $A \rightarrow B$ is a family of functions $L_X: D(A \otimes X) \rightarrow D(B \otimes X)$ such that:
    \begin{itemize}
    \item Parallel composition commutes with the action of $\mathscr{L}$, meaning  \[ \mathscr{L}_{XX'}(\rho_{AX} \otimes \sigma_{X'}) = \mathscr{L}_{X}(\rho_{AX}) \otimes \sigma_{X'} , \quad \quad \mathscr{L}_{X}(\rho_{A} \otimes \sigma_{X}) = \mathscr{L}_{[-]}(\rho_{A}) \otimes \sigma_{X} \]
        \item Probabilities of measurements on auxiliary systems are untouched by $\mathscr{L}$, meaning \[  {\rm Tr}(\mathscr{L}_{X}(\rho) (I \otimes \pi_{X})) =  {\rm Tr}(\rho (I \otimes \pi_{X}))  \]
        \item The update due to measurements on auxiliary systems are untouched by $\mathscr{L}$, meaning \[      \frac{(I \otimes \pi) \mathscr{L}_X (\rho) (I \otimes \pi)} {p(\mathscr{L}_X (\rho), I \otimes \pi)} = \mathscr{L}_{X}\Bigg(  \frac{ (I \otimes \pi) \rho (I \otimes \pi) }{p(\rho, I \otimes \pi)} \Bigg) , \] where we recall that we have defined $p(\rho , \pi ):= {\rm Tr}(\rho \pi) $.
    \end{itemize}
\end{example}
Again using the second and third bullet points together gives \[      \frac{(I \otimes \pi) \mathscr{L}_X (\rho) (I \otimes \pi)} {p(\rho, I \otimes \pi)} = \mathscr{L}_{X}\Bigg(  \frac{ (I \otimes \pi) \rho (I \otimes \pi) }{p(\rho, I \otimes \pi)} \Bigg) . \] Note that again, locally applicable transformations are only ever defined on the most easily physically interpretable positive linear operators, those which are normalised.

\begin{theorem} \label{thm:channel}
    $\mathscr{L}:A \rightarrow B$ is a locally applicable transformation on mixed quantum theory if and only if there exists a quantum channel $\mathcal{E}: A \rightarrow B$ such that $\mathscr{L}_X = \mathcal{E} \otimes I_X$ for all $X \in O$.\footnote{Again, strictly speaking, this theorem tells us that that $\mathscr{L}_X$ is the function obtained by restricting the domain of $\mathscr{E} \otimes I_X$ to trace-one density operators.}
\end{theorem}
\begin{proof}

    The proof is almost identical to the proof given for the pure case. 
    To prove that any quantum channel can be used to construct a locally-applicable transformation is straightforward, given for completeness in Lemma $2$ of the appendix. 
    Now let us see that this construction has an inverse, i.e that every locally-applicable transformation can be represented by a quantum channel. We adopt a ket/bra-like notation of completely positive maps between density operators. We use $\kett{\rho}$ to represent a density matrix and use $\braa{\sigma}$ to represent the effect associated to a density operator $\sigma$, that is, the quantum channel $\braa{\sigma}: H \rightarrow \mathbb{C}$ given by $\braa{\sigma}(\rho) = Tr[ \sigma \rho ] $. We denote $Tr[ \sigma \rho ]$ more cleanly in analogy to bra/ket notation as $\braa{\sigma} \kett{\rho}$, and we denote for some general channel $\mathcal{E}$ it's action on a density operator $\rho$ as simply $\mathcal{E}\kett{\rho}$. 
    
    One considers an arbitrary state $\rho_{AX}$, and sets up a teleportation-like experiment almost identical to the pure-case. More precisely, we consider the state  \[\mathscr{L}_{XA'X'}(\kett{\Phi^+}_{AA'} \otimes \kett{\Phi^+}_{XX'}) \otimes \kett{\rho}_{A''X''}.\] 
    and show that
    \begin{equation} \label{eq:linear}
        \mathscr{L}_{X}(\kett{\rho_{AX}}) = \mathcal{E}\kett{\rho}_{A''X''}
    \end{equation}
    where $\mathcal{E}: A'' \otimes X''\rightarrow B \otimes X$ is the quantum channel defined by
    \begin{equation}
        \mathcal{E} := d_A^2 d_X^2 (\mathcal{I}_{AX} \otimes \braa{\Phi^+}_{A'A''} \otimes \braa{\Phi^+}_{X'X''}) \mathscr{L}_{XA'X'}(\kett{\Phi^+}_{AA'} \otimes \kett{\Phi^+}_{XX'}) \otimes \mathcal{I}_{A''X''}.
    \end{equation}
    The proof is then follows identically to the pure case. In fact, the proof can be summarized diagrammatically using the same syntactic steps in Figure \ref{fig:proof}, with only the semantic interpretation changed to the theory of completely positive maps.

    It follows that $\mathscr{L}_X$ itself is a (restriction of the domain of a) completely positive linear map $\mathcal{E}$. To show that $\mathcal{E}$ is furthermore trace-preserving and therefore a quantum channel, it suffices to note that since all locally applicable transformations are defined as functions with domain and codomain given by the trace-one density operators. Consequently we have that, 
    \begin{equation}
       \forall \rho : Tr[\rho] = 1 \ \text{then} \  \mathcal{E}(\rho) =  Tr[\mathscr{L}_X (\rho)] = 1,
    \end{equation}
    and so $\mathcal{E}$ is indeed trace-preserving. Finally, state locality implies that for any $X \in O$, we have $\mathscr{L}_X =  \mathscr{L}_{[-]}\otimes I_X$. 
\end{proof}

Consequently, there is a composition-preserving, one-to-one correspondence between the locally applicable transformations on mixed quantum SM theory and the quantum channels (that is, an equivalence of categories). 
The postulate of local applicability is logically equivalent to the postulate of complete positivity and trace preservation for the transformations of mixed quantum theory.

\section{Comparison with the Argument of Gisin} \label{sec:gisin}

In \cite{Gisin:1989sx_OG_nosig_paper}, Gisin addressed the question of whether quantum dynamics are a consequence of relativity. Defining the Schrödinger evolution as the rule that states evolve as
\begin{equation} \label{eq:Schrödinger}
    \ket{\psi(t)} = e^{-iHt/\hbar} \ket{\psi(0)}
\end{equation} 
for some Hamiltonian $H$, and assuming the projection postulate, the abstract of \cite{Gisin:1989sx_OG_nosig_paper} includes the following statement:
\begin{quote}
\textit{the Schrödinger evolution is the only quantum evolution that is deterministic and compatible with relativity.}
\end{quote}

While definitions of each individual term are not explicitly given, it seems reasonable to infer that  `quantum evolution' refers to something like a family of functions on quantum states parameterized by time\footnote{This is suggested by the presentation of the Schrödinger evolution, and later reference in the paper to QSDs, although when defining a deterministic evolution, a single function is studied.}; `deterministic' refers to a mapping from pure states to pure states\footnote{See the bottom of p.364, and p.366 where it is suggested that any nondeterministic evolution would map a pure state to a mixed state.}; and `compatible with relativity' means that the evolution does not lead to superluminal signaling.

When we try to explicitly bring the argument to its  full conclusion we will find that we need some additional assumption(s). Before getting there, let us first recap the argument as given in \cite{Gisin:1989sx_OG_nosig_paper}, which elegantly recovers convex linearity on density matrices from compatibility with relativity.

\subsection{A review of the argument}

By determinism, the dynamics are represented by a function $g$ on pure density operators. This function induces a mapping on probabilistic ensenbles of pure density operators, 
\begin{equation}
    \{p_x, \ket{\psi^{(x)}}\bra{\psi^{(x)}}\} \quad \mapsto \quad \{p_x, g(\ket{\psi^{(x)}}\bra{\psi^{(x)}})\}.
\end{equation}
Now let us introduce some terminology. $g$ is called \textit{convex linear} if
\begin{equation} \label{eq:convexlinear}
\begin{split}
        \sum_x  p_x \ket{\psi^{(x)}}\bra{\psi^{(x)}} 
        &= \sum_x q_x\ket{\phi^{(x)}} \bra{\phi^{(x)}} \\ \implies \sum_x p_x g(\ket{\psi^{(x)}}\bra{\psi^{(x)}}) &= \sum_x q_x g(\ket{\phi^{(x)}} \bra{\phi^{(x)}}) 
\end{split}
\end{equation}
for any two probabilistic ensembles of $\{p_x, \ket{\psi^{(x)}}\bra{\psi^{(x)}}\}$ and $\{p_x, \ket{\phi^{(x)}}\bra{\phi^{(x)}}\}$. A pair of ensembles satisfying $ \sum_x p_x \ket{\psi^{(x)}}\bra{\psi^{(x)}}  = \sum_x \ket{\phi^{(x)}}\bra{\phi^{(x)}}$ are called \textit{indistinguishable}. 

Suppose that $g$ is not convex linear, i.e.\ does not satisfy (\ref{eq:convexlinear}). Then there exists an indistinguishable pair of ensembles $(e_1, e_2)$ that are mapped to a distinguishable pair $(e_1', e_2')$. Assuming that quantum states and measurements are treated in the usual way, we can now deploy a useful lemma from \cite{Gisin:1989sx_OG_nosig_paper}. This lemma shows that for any indistinguishable $(e_1, e_2)$, one can devise a situation in which Alice chooses between two measurements, and, in doing so, chooses whether Bob's spacelike separated system is prepared in $e_1$ or $e_2$. If Bob subsequently implements $g$, then Alice effectively chooses between $e_1'$ or $e_2'$. But then the distinguishability of $(e_1', e_2')$ implies that Alice can send signals to Bob. In short, the usual treatment of states and measurements in pure quantum SM theory implies that any convex nonlinear dynamical function leads to superluminal signaling.

After this derivation of convex linearity, a further argument is not given (or cited) as to why the Schrödinger evolution (\ref{eq:Schrödinger}) is the only quantum evolution that is deterministic and compatible with relativity. 


\subsection{Non-linear functions with no-superluminal signalling}

In fact, this conclusion cannot be reached without additional assumptions. To see this, consider the following function, where $\ket{0}\bra{0}$ is some fixed state.
\begin{equation}
    g(\ket{
    \psi}\bra{
    \psi}) = \ket{0}\bra{0}
\end{equation}
$g$ is deterministic (in the sense of mapping pure states to pure) and compatible with relativity (in the sense of not leading to superluminal signalling). But it does not arise from any Schrödinger evolution.

Indeed, $g$ does not even act linearly on the Hilbert space -- there is no linear operator $L$ such that $g(\ket{\psi}\bra{\psi})=L\ket{\psi}\bra{\psi}L^\dag$ (see Appendix \ref{app:nonlinear}).
As the argument from \cite{Gisin:1989sx_OG_nosig_paper} shows, $g$ is \textit{convex linear}. But convex linearity does not imply linearity with respect to superpositions. This is why the determinism and no-signaling assumptions of \cite{Gisin:1989sx_OG_nosig_paper} do not get us all the way to unitary linear evolution without additional assumptions. 

With this in mind, and as noted in \cite{bassi2003dynamical}, one way to complete the argument is to invoke a theorem from \cite{davies1976quantum}. First note that any convex linear map $g$ on {pure density operators} defines a positive linear map on all density operators. Then
Theorem 3.4 of \cite{davies1976quantum} tells us that that any dynamical \textit{group} -- a semigroup where each element has an inverse -- of positive linear maps on density operators is generated by the Schrödinger evolution. So if we assume not only that transformations are deteminisitic and compatible with relativity theory, but also that they form a dynamical group, then the Schr\"odinger evolution can be derived. In short, 
\begin{quote}
Determinism + Compatibility with relativity + Dynamical groups $\implies$ Schrödinger evolution.
\end{quote}

The dynamical group assumption appears to bring three conceptual additions to the argument, the existence of a continuous time-parameter, the reversibility of transformations, and Markovianity.

Perhaps the reversibility assumption could be dropped: as far as we know, it is possible that the theorem above will still hold if we replace `groups' with `semigroups'. One could then even say that the quote form the start of this section is true as long as forming a dynamical semigroup is a thought of as a defining property of an evolution.
This would be an interesting and valuable result which would clarify the role of signalling constraints in deriving linearity, well worth either proving or disproving in future work.

If a proof can be given without reversibility, one still has to assume that one has a continuous time parameter, and, moreover, a dynamical semigroup, before one can derive linearity with respect to superpositions and unitarity\footnote{Note that this also means that it is unclear how a Gisin style argument could be used to derive the linearity of the (isometric) transformations between distinct Hilbert spaces, meaning there is a whole class of kinds of transformations of pure states who's linearity is so-far only derivable from local applicability}. {On the other hand, our derivation of the linearity and unitarity of pure quantum transformation makes no such assumptions. Even if time is fundamentally discrete, all locally applicable transformations on pure quantum SM theory are unitary. }


\subsection{On the Argument for Complete Positivity}
Another nice feature of the local applicability axiom is that it can be applied without tweaking to recover the linear unitary maps in the pure setting and the quantum channels in the mixed setting. In \cite{Simon_2001_gisin_CPTP}, it is suggested that a Gisin-style no superluminal signalling by dynamics of density matrices argument could be used to recover the quantum channels. One might wonder then, whether a Gisin-style argument, when assisted with dynamical group conditions, could be used in a similar way. In trying to do so we find there are two key challenges.

First, the argument goes as far as to show that any function from pure states to density matrices which satisfies no superluminal signalling can be extended to a convex linear map on the density matrices. Positivity is clear by construction, however, less clear is complete positivity. Indeed, from this point onwards, complete positivity is axiomatised in the usual way as an additional property for positive linear maps. It is unclear how to phrase such a requirement in a theory-independent way, without prior knowledge of linearity.

Second, without relaxation of the dynamical group requirement in the pure setting to a dynamical semi-group requirement, the assistance of reversibility used to complete the argument in the pure setting is then too strong in the mixed setting, ruling out all quantum channels which are not isometric.

\subsection{Summary of Comparison}
The approach of \cite{Gisin:1989sx_OG_nosig_paper} provides an enlightening and efficient derivation of convex linearity, something more needs to be said however before quantum linearity and unitarity are recovered. That `something more' at the very least prevents the argument from applying to discrete time evolution. Additionally, it seems to be not possible to put together the arguments of \cite{Gisin:1989sx_OG_nosig_paper, Simon_2001_gisin_CPTP}, in such a way as to give a simple one-size-fits-all axiom for transformations.

This highlights a couple of nice features of local applicability. First, it recovers linearity and unitarity without background assumptions of reversibility or any a particular notion of time. Second, it is a working axiom for quantum transformations, which works without tweaking simultaneously for the pure and mixed settings.

\section{Conclusion} 

On a standard approach to quantum theory, one starts by making the assumption that the dynamics of an isolated system can be represented by a linear and unitary operator on its Hilbert space. 
Then, one further assumes that the action of $U$ in the presence of an environment is given by $U \otimes I$. Once both these assumptions are made, it follows that $U$ is locally applicable. 

But the argument can be reversed  -- one can \textit{assume} local applicability and then 
\textit{derive} linearity, unitarity, and the role of $U \otimes I$ (or, in the mixed case, convex linearity, completely positivity and $\ce \otimes \ci$). This flipping of the script not only makes for a more intuitive formulation of quantum theory; it also throws light on its connection with the other pillar of modern physics, relativity theory.

In connecting locality with linearity, the results of this paper also connect locality with core foundational issues in quantum theory.
In particular, linearity with respect to superpositions is responsible for the Wigner's friend paradox \cite{wigner1995remarks} and the associated problems of quantum measurements. Theorem \ref{thm:unitary} thus reveals a sense in which locality leads to a measurement problem. This complements insights from recent no-go theorems \cite{healey2018quantum, bong2020strong, haddara2022possibilistic,  leegwater2022greenberger, ormrod2022no} stating that the combination of quantum theory and relativity (and more generally, Bell nonlocality, the preservation of information, and dynamical locality \cite{ormrod2023theories}) leads to a particularly acute sort of measurement problem, in which the observed outcomes and settings of measurements fail to be absolute. 

On the mathematical side, some connection appears to be forming between the foundations of physics and the Yoneda lemma, a fundamental theorem of abstract algebra and category theory \cite{Lane1971CategoriesMathematician}. This lemma states that families of functions (suitably well-behaved and called natural transformations) will always inherit the defining structural features of the mathematical theories they act on. In this paper, we have seen that whichever form of linearity was present in a state space (be it the superposition rules of Hilbert spaces or the mixing rules of operator spaces) was inherited by families of functions (suitably well-behaved and called locally-applicable transformations). This suggests the possibility of a version of the Yoneda lemma specially adapted for physical theories. 

In order to derive our results, we were led to introduce the notion of a \textit{state-measurement theory}. In the future, this very abstract framework could be used much more generally to study what dynamics are imposed on a theory either by local applicability or by other assumptions. Most obviously, one could attempt to extend our theorems to infinite-dimensional quantum state-measurement theories, or to state-measurement theories associated with other operational probabilistic theories \cite{Chiribella_2010_prob_pure}. One could also attempt an extension to theories with modified Born rules \cite{aaronson2004quantum, Bao_2016, born_causality, Galley2018anymodificationof, Galley2017classificationofall, Masanes_2019}. Finally, one could attempt to derive features of a given theory besides its dynamics, including the tensor product rule, the direct sum, and even the state-measurement rule itself. If local applicability is not enough for all of this, there is a natural question: what other assumption(s) will be?

It is expected in quantum geometries that local subsystems will be represented with more subtle mathematical tools than simple tensor products, such as restrictions \cite{arrighi2022quantum} and non-factor sub-algebras \cite{Giddings_2006,  Bianchi_2019, Giddings_2015, Chiribella_2018}. One could attempt to derive the dynamics in these contexts by considering something like local applicability, but defined in terms of \textit{direct sum} structures rather than (just) in terms of tensor products. Given the problems associated with time in quantum gravity (including its possible discreteness, the lack of a global time parameter, and superpositions of time), it is encouraging that our derivation of unitarity did not rely on having any time parameter (continuous or otherwise).

Finally, in very broad terms, this is not the first context in which local applicability has been used to generalise, and reconstruct, standard notions of transformation: \cite{wilson2023quantum, wilson2022free} did something similar in the more specialised study of quantum supermaps. This leads us to wonder in which other settings this simple physical principle can be formalised and used to reconstruct and generalise accepted notions of transformation in science and mathematics.

\section*{Acknowledgements} We are pleased to thank Kai-isaak Ellers, James Hefford, and Jonathan Barrett for useful conversations, and the organisers of the QISS Spring School 2023 for a thoroughly enjoyable week that stimulated the development of this paper.
This work was funded by the Engineering and Physical Sciences Research Council [grant number EP/W524335/1] and [grant number EP-T517823-1]. 

\bibliographystyle{utphys.bst}
\bibliography{locality_refs}

\appendix

\section{Using the Tensor Product to Construct Standard Transformations}
As mentioned in the main text, unitary linear maps always define locally-applicable transformations of pure states, and quantum channels always define locally-applicable transformations of mixed states. Let us begin by proving the former statement.


\begingroup
  \setcounter{lemma}{0} 
  \begin{lemma}[Unitary Linear Maps] \label{example:unitary}
    Let \( U: A \to B \) be an isometric linear map between Hilbert spaces, 
    a locally applicable transformation \( \mathscr{L}^{U}:A \rightarrow B \) 
    can be constructed on pure quantum SM theory by defining
    \[
      \mathscr{L}^{U}_X(\psi) :=  U \otimes I_{X} \ket{\psi}.
    \]
  \end{lemma}
\endgroup
\begin{proof} Clearly $\mathscr{L}_X$ maps normalized vectors to normalized vectors for all $X$. Now let us check that each part of local applicability holds for $\mathscr{L}^U$, beginning with state locality. For any $\ket{\psi} \in S_{A  X}$ and $\ket{\phi} \in S_{X'}$:
\begin{align*}
    \mathscr{L}^{U}_{XX'}(\ket{\psi} \otimes \ket{\phi})  & =   (U \otimes I_{XX'}) \ket{\psi} \otimes \ket{\phi} \\
    & = (U \otimes I_{X}) \ket{\psi} \otimes \ket{\phi} \\
    & = \mathscr{L}^U_{X}(\ket{\psi}) \otimes \ket{\phi}.
\end{align*} We now verify that no signaling is satisfied.
\begin{align*}
\mathscr{L}^U_X (\ket{\psi})^\dag  I \otimes \pi \ket{\mathscr{L}^U_X \psi} & =  \bra{\psi} (U^{\dagger} \otimes I) (I \otimes \pi) (U \otimes I) \ket{\psi} \\
 & = \bra{\psi} (I \otimes \pi)  (U^{\dagger} \otimes I) (U \otimes I) \ket{\psi} \\
 & = \bra{\psi}  (I \otimes \pi) \ket{\psi}
\end{align*} Finally, we check update commutativity.
\begin{align*}
 \frac{(I \otimes \pi) \ket{\mathscr{L}^U_X (\psi)}}{\sqrt{p(\ket{\psi}, I \otimes \pi)}} & = \frac{(I \otimes \pi) (U \otimes I) \ket{\psi}}{\sqrt{p(\ket{\psi}, I \otimes \pi)}} \\
  & = \frac{(U \otimes I) (I \otimes \pi) \ket{\psi}}{\sqrt{p(\ket{\psi}, I \otimes \pi)}} \\
  & = \mathscr{L}^U_{X}  \Bigg(\frac{ (I \otimes \pi) \ket{\psi} }{\sqrt{p(\ket{\psi}, I \otimes \pi)}} \Bigg).
 \end{align*}
\end{proof}

Let us now move on to the mixed setting. 
\begin{lemma}[Quantum Channels] \label{lem:unitary}
Let $\mathcal{E}:A \rightarrow B$ be a completely positive trace-preserving map between Hilbert spaces, a locally applicable transformation $\mathscr{L}_{\mathcal{E}}:A \rightarrow B$ can be constructed on mixed quantum theory by defining \[  \mathscr{L}^{\mathcal{E}}_X(\rho) :=  \mathcal{E} \otimes \mathcal{I}_{X} [\rho].   \] 
\end{lemma}
\begin{proof}
It is clear that whenever $\rho$ is normalised then so is $\mathscr{L}^{\mathcal{E}}(\rho)$. We now check state locality.
\begin{align*}
    \mathscr{L}^{\mathcal{E}}_{XX'}(\rho_{AX} \otimes \sigma_{X'})  & =   (\mathcal{E} \otimes \mathcal{I}_{XX'}) [\rho_{AX} \otimes \sigma_{X'}] \\
    & = (\mathcal{E} \otimes \mathcal{I}_{X}) [\rho_{AX}] \otimes \sigma_{X'}  \\
    & = \mathscr{L}^{\mathcal{E}}_{X}(\rho_{AX}) \otimes \sigma_{X'}.
\end{align*} Next, we confirm that no signalling holds.
\begin{align*}
{\rm Tr}(\mathscr{L}^{\mathcal{E}}_{X}(\rho) (I \otimes \pi_{X})) & = {\rm Tr}((\mathcal{E} \otimes \mathcal{I}_{X})[\rho] (I \otimes \pi_{X}))  \\
& = {\rm Tr}(\rho (I \otimes \pi_{X})) 
\end{align*} Finally, we check update commutativity.
\begin{align*}
 \frac{(I \otimes \pi) \mathscr{L}^{\mathcal{E}}_X (\rho) (I \otimes \pi)} {p(\rho, I \otimes \pi)} & =  \frac{(I \otimes \pi) (\mathcal{E} \otimes \mathcal{I}_{X}) [\rho] (I \otimes \pi)} {p(\rho, I \otimes \pi)}  \\
 & = (\mathcal{E} \otimes \mathcal{I}_{X})\Bigg[ \frac{(I \otimes \pi)  \rho (I \otimes \pi)} {p(\rho, I \otimes \pi)} \Bigg] \\
& = \mathscr{L}^{\mathcal{E}}_{X}\Bigg(  \frac{ (I \otimes \pi) \rho (I \otimes \pi) }{p(\rho, I \otimes \pi)} \Bigg)
 \end{align*} This completes the proof. 
\end{proof}

\section{A nonlinear map that is deterministic and compatible with relativity} \label{app:nonlinear}

In Section \ref{sec:gisin}, we introduced a function $g(\ket{\psi}\bra{\psi}) := \ket{0}\bra{0}$ that is deterministic and compatible with relativity in Gisin's sense. Here, we show that this function acts nonlinearly on the Hilbert space. More precisely, we show that there is no linear operator $L$ the Hilbert space such that
\begin{equation} \label{eq:impossible}
    L\ket{\psi}\bra{\psi}L^\dag = \ket{0}\bra{0}
\end{equation}
for all states $\ket{\psi}$ and some specific state $\ket{0}$. 

Let us define $\ket{+}:= \frac{\ket{0}+\ket{1}}{\sqrt{2}}$ for some state $\ket{1}$ orthogonal to $\ket{0}$. If there were an operator satisfying (\ref{eq:impossible}), then it would have to satisfy
\begin{equation} \nonumber
\begin{split}
    \ket{0} \bra{0} &=   L\ket{+}\bra{+}L^\dag \\
    &= \frac{L\ket{0}\bra{0}L^\dag +L\ket{0}\bra{1}L^\dag +L\ket{1}\bra{0}L^\dag +L\ket{1}\bra{1}L^\dag}{2} \\
    &= \ket{0}\bra{0} + \frac{L\ket{0}\bra{1}L^\dag +L\ket{1}\bra{0}L^\dag}{2} 
\end{split}
\end{equation}
meaning that $\frac{L\ket{0}\bra{1}L^\dag +L\ket{1}\bra{0}L^\dag}{2} =0$. Defining $\ket{+i}:= \frac{\ket{0}+i\ket{1}}{\sqrt{2}}$ for some state $\ket{1}$, we also have
\begin{equation}  \nonumber
\begin{split}
    \ket{0} \bra{0} &=   L\ket{+i}\bra{+i}L^\dag \\
    &= \frac{L\ket{0}\bra{0}L^\dag -i L\ket{0}\bra{1}L^\dag + i L\ket{1}\bra{0}L^\dag +L\ket{1}\bra{1}L^\dag}{2} \\
    &= \ket{0}\bra{0} + i\frac{-L\ket{0}\bra{1}L^\dag +L\ket{1}\bra{0}L^\dag}{2} 
\end{split}
\end{equation}
and so $\frac{-L\ket{0}\bra{1}L^\dag +L\ket{1}\bra{0}L^\dag}{2} =0$. It follows that $\L\ket{0}\bra{1}L^\dag=0$, implying that either $L\ket{0}=0$ or $L\ket{1}=0$, and ultimately that either $L\ket{0}\bra{0}L^\dag=0$ or $L\ket{1}\bra{1}L^\dag=0$, in violation of (\ref{eq:impossible}).

\end{document}